\documentclass[final,10pt,times,authoryear]{elsarticle}
\usepackage{amsmath}
\usepackage{amsthm}
\usepackage{amssymb}
\usepackage{amsfonts}
\usepackage[utf8]{inputenc}

\newtheorem{theorem}{Theorem}
\newtheorem{lemma}[theorem]{Lemma}
\newtheorem{corollary}[theorem]{Corollary}

\newtheorem{proposition}[theorem]{Proposition}
\newtheorem{property}[theorem]{Property}
\newtheorem{definition}[theorem]{Definition}

\newtheorem{remark}[theorem]{Remark}

\newcommand{\OpenDreamKit}{the \href{http://opendreamkit.org}{OpenDreamKit} \href{https://ec.europa.eu/programmes/horizon2020/}{Horizon 2020} \href{https://ec.europa.eu/programmes/horizon2020/en/h2020-section/european-research-infrastructures-including-e-infrastructures}{European Research Infrastructures} project (\#\href{http://cordis.europa.eu/project/rcn/198334_en.html}{676541})}

\usepackage{booktabs}
\usepackage[plainpages=true]{hyperref}
\hypersetup{
  pdftitle={Time and space efficient generators for quasiseparable matrices},
  pdfauthor={Clément Pernet and Arne Storjohann},
  breaklinks=true, 
  colorlinks=true,
  linkcolor=darkred,
  citecolor=blue,
  urlcolor=darkgreen,
}
\usepackage{graphicx}
\graphicspath{{Pictures/}}
\newlength{\myfigsize}
\usepackage{algorithm,algpseudocode}
\usepackage{xspace}
\usepackage{array}
\usepackage{multirow}
\usepackage{color,svgcolor}

\newenvironment{smatrix}{\left[\begin{smallmatrix}}{\end{smallmatrix}\right]}

\newcommand{\mat}[1]{\mathsf{#1}\xspace}
\newcommand{\RPM}{\ensuremath{\mathcal{R}}\xspace}
\newcommand{\LTP}{\ensuremath{\text{Left}}\xspace}
\newcommand{\GO}[1]{\ensuremath{O(#1)}\xspace}
\newcommand{\SO}[1]{\ensuremath{O\tilde\ (#1)}\xspace}

\begin{document}

\begin{frontmatter}
  \title{Time and space efficient generators for quasiseparable matrices}
\tnotetext[t1]{This work is partly funded by \OpenDreamKit.}
  \author{Clément Pernet}
  \address{Université Grenoble Alpes, Laboratoire Jean Kuntzmann, \\ ÉNS de Lyon, Laboratoire de l'Informatique du
    Parall\'elisme.}
  \ead{clement.pernet@imag.fr}
  \ead[url]{http://ljk.imag.fr/membres/Clement.Pernet/}
  \author{Arne Storjohann}
  \address{David R. Cheriton School of Computer Science\\ University of Waterloo,
    Ontario, Canada, N2L 3G1}
  \ead{astorjoh@cs.uwaterloo.ca}
  \ead[url]{https://cs.uwaterloo.ca/~astorjoh/}
  \begin{abstract}
  The class of quasiseparable matrices is defined by the property that any
  submatrix entirely below or above the main diagonal has small rank, namely
  below a bound called the order of quasiseparability. These matrices
  arise naturally in solving PDE's for particle interaction with the Fast
 Multi-pole Method (FMM), or computing generalized eigenvalues. From these
 application fields, structured
 representations and algorithms have been designed in numerical linear algebra
 to compute with these matrices in time linear in the matrix dimension and
 either quadratic or cubic in  the quasiseparability order.
 Motivated by the design of the general purpose exact linear algebra library
 LinBox, and by algorithmic applications in algebraic computing, we adapt
 existing techniques introduce novel ones to use quasiseparable matrices in
 exact linear algebra, where sub-cubic matrix arithmetic is available.
 In particular, we will show, the connection between the notion of
 quasiseparability and the rank profile matrix invariant, that we have
 introduced in 2015. It results in two new structured representations, one being
 a simpler variation on the hierarchically semiseparable storage, and the second
 one exploiting the generalized Bruhat decomposition.
As a consequence, most basic operations, such as computing the quasiseparability
orders, applying a vector, a block vector, multiplying two quasiseparable
matrices together, inverting a quasiseparable matrix, can be at least as fast
and often faster than previous existing algorithms.

  \end{abstract}
\end{frontmatter}

\section{Introduction}

We consider the class of quasiseparable matrices, defined by a bounding
condition on the ranks of the 
submatrices in their lower and upper triangular parts. These structured matrices
originate mainly from two distinct application fields: computing generalized
eigenvalues~\citep{GKK85,EiGo99}, and solving
partial differential equations for particule simulation with the
fast multipole method~\citep{CGR88}. This class also arise naturally, as it includes the
closure under inversion of the class of banded matrices.
Among the several definitions used in the litterature, we will use that
of~\cite{EiGo99} for the class of quasiseparable matrices.
\begin{definition}\label{def:quasisep}
  An $n\times n$ matrix $\mat{M}$ is $(r_L,r_U)$-quasi\-se\-pa\-ra\-ble if its
  strictly lower and upper  triangular parts satisfy the following low rank structure:
  for all $1\leq k\leq n-1$,
  \begin{eqnarray}
    \text{rank}(\mat{M}_{k+1..n,1..k})&\leq& r_L, \label{eq:quasisep:low}\\
    \text{rank}(\mat{M}_{1..k,k+1..n})&\leq& r_U \label{eq:quasisep:up}.
  \end{eqnarray}
The values $r_L$ and $r_U$ 
are the quasiseparable orders of $\mat{M}$.
\end{definition}

Other popular classes of structured matrices like
Toeplitz, Vandermonde, Cauchy, Hankel matrices and their block versions, enjoy a
unified description through the powerful notion of displacement
rank~\citep{KKM79}. Consequently they benefit from space efficient
representations (linear in the dimension $n$ and in the displacement rank~$s$),
and time efficient algorithms to apply them to a vector, compute their inverse
and solve linear systems: most operations have been reduced to polynomial
arithmetic~\citep{Pan90,BiPa94}, and by incorporating fast matrix algebra, this cost
has been reduced from $\SO{s^2n}$ to $\SO{s^{\omega-1}n}$ by~\cite{BJS08} (assuming
that two $n\times n$ matrices can be mutliplied in $\GO{n^\omega}$ for
$2.3728639\leq \omega \leq 3$~\citep{LeG14}).

However quasiseparable matrices do not fit in the framework of rank displacement
structures.
Taking advantage of the low rank properties, mainly two types of structured
representations have been developped together with corresponding dedicated
algorithms to perform common linear algebra operations:
the quasiseparable generators of~\cite{EiGo99,VBGM05,VVM07}, their
generalization for finite block matrices by~\cite{EiGo05}, that coincides with the
sequentially semiseparable (SSS) representation of~\cite{CDGPSVW05} and  
the hierarchically semiseparable representations (HSS) of~\cite{CGP06,XCSGL10}.
We refer to~\citep{VBGM05},~\citep{VVM07} and~\cite{XCSGL10} for a broad bibliographic overview on
 the topic. Note also the alternative approach of Givens and unitary weights in~\cite{DVB07}.


\paragraph{Sequentially Semiseparable representation} The sequentially semiseparable
representation used by~\cite{EiGo99,VBGM05,VVM07,EGO05,BEG16} for 
a matrix $\mat{M}$, consists of $(n-1)$ pairs of
vectors $\mat{p}(i),\mat{q}(i)$ of size $r_L$, $(n-1)$ pairs of vectors $\mat{g}(i),\mat{h}(i)$ of size
$r_U$,  $n-1$ matrices $\mat{a}(i)$ of dimension $r_L\times r_L$, and  $n-1$
matrices $\mat{b}(i)$ of dimension $r_U\times r_U$ and $n$ scalars $d(i)$ such that 
$$
\mat{M}_{i,j} = \left\{
  \begin{array}{ll}
    \mat{p}(i)^T\mat{a}^{>}_{ij} \mat{q}(j), & 1\leq j<i\leq n\\
    d(i),& 1\leq i=j\leq n\\
    \mat{g}(i)^T\mat{b}^{<}_{ij} \mat{h}(j), & 1\leq i<j\leq n\\
  \end{array}
\right.
$$
where
\begin{align*}
  \mat{a}^{>}_{ij} = \mat{a}(i-1)\dots \mat{a}(j+1) &\text{ for } j>i+1, &
  \mat{a}_{j+1,j}&=1,\\
  \mat{b}^{<}_{ij} = \mat{b}(i+1)\dots \mat{b}(i-1) &\text{ for } i>j+1, &
  \mat{b}_{i,i+1}&=1.
  \end{align*}
For $s=\max(r_L,r_U)$, this representation, of size $\GO{n(r_L^2+r_U^2)} =
\GO{s^2n}$ makes it possible to apply a 
  vector in $\GO{s^2n}$ field operations, multiply two quasiseparable
  matrices in time $\GO{s^3n}$ and also compute the inverse of a strongly
  regular matrix in time $\GO{s^3n}$~\citep{EiGo99}. Note that the inefficiency in size
  for these represention can be mitigated using the blocked version of this
  representation of~\cite{EiGo05}.

  \paragraph{The Hierarchically Semiseparable representation}
The Hierarchically Semiseparable representation was introduced in~\cite{CGP06}
and is related to the structure used in the Fast Multipole Method~\citep{CGR88}.
It is based on the splitting of the matrix in four quadrants, the use of rank
revealing factorizations of its off-diagonal quadrants and applying the same
scheme recursively on the diagonal blocks. A further compression is applied to
represent all off-diagonal blocks as linear combinations (called translation
operators) of blocks of a finer recursive order.
While the space and time complexity of the HSS representation is depending on numerous
parameters, the analysis in~\cite{CGP06} seem to indicate that the size of an
HSS representation is $\GO{sn}$, it can be applied to a vector in linear time in its
size, and linear systems can be solved in $\GO{s^2n}$. For the product of two
HSS matrices, we could not find any better estimate than $\GO{s^3n}$ given
by~\cite{SDC07}.

  
\paragraph{Context and motivation}
The motivation here is to propose simplified and improved representations of
quasiseparable matrices (in space and time). Our approach does not focus on
numerical stability for the moment. Our first motivation is indeed to use these
structured matrices in computer algebra where computing e.g. over a finite field
or over  multiprecision integers and rationals does not lead to any numerical
instability. Hence we will assume throughout
the paper that any Gaussian elimination algorithm mentioned has the ability to
reveal ranks. In numerical linear algebra, a special care need to be taken for the
pivoting of LU decompositions~\citep{HWY92,Pan00}, and QR or SVD decompositions are often
preferred~\citep{Chan87,ChIp94}. Part of the methods presented here,
namely that of Section~\ref{sec:RRR}, rely on an arbitrary rank revealing matrix
factorization and can therefore be applied to a setting with numerical instability. 
In the contrary, Section~\ref{sec:CB} relies on a class of Gaussian elimination
algorithm that reveal the rank profile matrix, hence applying it to numerical
setting is future work. 
This study is motivated by the design of new algorithms on polynomial matrices
over a finite field, where
quasiseparable matrices naturally occur, and  more generally by the framework of the
\texttt{LinBox} library~\citep{linbox:2016} for black-box exact linear algebra. 

\paragraph{Contribution}
This paper presents in further details and extends the results of~\cite{Per16},
while also fixing a mistake~\footnote{Equation~(9) in~\cite{Per16} is missing
  the \LTP operators. The resulting algorithms are incorrect. This is fixed in section~\ref{sec:CBxTS}.}.
It proposes two new  structured representations for quasiseparable matrices, a Recursive
Rank Revealing (RRR) representation that can be viewed as a simplified version of the HSS
representation of~\cite{CGP06}, and a representation based on the generalized Bruhat
decomposition, which we name Compact Bruhat (CB) representation.
The later one, is made possible by the connection that we make between the
notion of quasiseparability and a matrix invariant, the rank profile matrix,
that we introduced in~\cite{DPS15} and applied to the generalized Bruhat
decomposition in~\cite{DPS16}. More precisely, we show that the lower and upper
triangular parts of a quasiseparabile matrix have a Generalized Bruhat
decompositions off of which many coefficients can be shaved. The resulting
structure of these decompositions allows to handle them within memory footprint
and time complexity that does not depend on the rank but on the quasiseparable
order (which can be arbitrarily lower).
These two representations use respectively a space \GO{sn\log\frac{n}{s}} (RRR) and \GO{sn} (CB),
hence improving over that of the SSS, \GO{s^2n}, and matching that
of the HSS representation, \GO{sn}.

The complexity of applying a vector remains linear in the size of the
representations. The main improvement in these two representations is in the complexity of
applying them to matrices and computing the matrix inverse, where we replace by $s^{\omega-1}$ the $s^3$ factor of
the SSS or the $s^2$ factor of the HSS
representations.\footnote{Note that most complexities for SSS and HSS in the 
  litterature are given in the form \GO{n^2} or \GO{n}, considering the
  parameter $s$ as a constant. The estimates given here, with the exponent in
  $s$, can be found in the proofs of the related papers or easily derived from
  the algorithms.} 
Table~\ref{tab:contrib}
\begin{table}[htbp]
  \begin{center}
      \begin{tabular}{lllll}
    \toprule
      & SSS & HSS & RRR & CB\\
    \midrule
    Size         & \GO{s^2n} & \GO{sn}  & \GO{sn\log\frac{n}{s}} & \GO{sn}\\
    Construction & \GO{s^2n^2} & \GO{sn^2} & \GO{s^{\omega-2}n^2} & \GO{s^{\omega-2}n^2} \\
    QSxVec & \GO{s^2n} & \GO{sn} & \GO{sn\log\frac{n}{s}} & \GO{sn}\\
    QSxTS  & \GO{s^3n} & \GO{s^2n} & \GO{s^{\omega-1}n\log\frac{n}{s}} & \GO{s^{\omega-1}n}\\
    QSxQS  & \GO{s^3n} & \GO{s^3n} & \GO{s^{\omega-1}n\log^2\frac{n}{s}} & \GO{s^{\omega-2}n^2}\\
    LinSys & \GO{s^3n} & \GO{s^2n} & \GO{s^{\omega-1}n\log^2\frac{n}{s}} & \\
    \bottomrule
      \end{tabular}     
  \caption{Comparing the size and time complexities for basic operations of the
    proposed RRR and CB representations with the existing one SSS and
    HSS on an
    $n\times n$ quasiseparable matrix of order $s$.}\label{tab:contrib}
  \end{center}
\end{table}
compares the two proposed representations with the SSS and the HSS in
their the size, and the complexity of the main basic operations.

\paragraph{Outline}
Section~\ref{sec:prelim} defines and recalls some preliminary notions on left
triangular matrices and the rank profile matrix, that will be used in
Section~\ref{sec:comporders} and~\ref{sec:CB}.
Using the strong connection between the rank profile matrix and the
quasiseparable structure, we first propose in Section~\ref{sec:comporders} an
algorithm to compute the quasiseparability orders $(r_L,r_U)$ of any dense
matrix in \GO{n^2s^{\omega-2}} where $s=\max(r_L,r_U)$.
Section~\ref{sec:generators} then describes the two proposed structured
representations for quasiseparable matrices:
the Recursive Rank Revealing representation (RRR), a simplified HSS representation
based on a binary tree of rank revealing factorizations, and the Compact Bruhat
representation (CB), based on the intermediate Bruhat representation.
Section~\ref{sec:RRR} then presents algorithms to compute an RRR representation,
and perform the most common operations with it: applying a vector, a tall and
skinny matrix, multiplying two quasiseparable matrices in RRR representation,
and computing the inverse of a strongly regular RRR matrix.
Section~\ref{sec:CB} presents algorithms to compute a Compact Bruhat
representation, and multiply it with a vector, a tall and skinny matrix or a dense matrix.

\paragraph{Notations}
Throughout the paper, $\mat{A}_{i..j,k..l}$ will denote the sub-matrix of
$\mat{A}$ of row indices between $i$ and $j$ and column indices between $k$ and $l$.
The matrix of the canonical basis, with a one at position $(i,j)$ will be
denoted by~$\mat{\Delta}^{(i,j)}$. We will denote the identity matrix of order $n$ by
$\mat{I}_n$, the unit antidiagonal of dimension $n$ by $\mat{J}_n$ and the zero
matrix of dimension $m\times n$ by $\mat{0}_{m\times n}$.

\section{Preliminaries}\label{sec:prelim}
\subsection{Left triangular matrices}
We will make intensive use of matrices with non-zero elements only located above the main anti-diagonal.
We will refer to these matrices as left triangular, to avoid any confusion with
upper triangular matrices.
\begin{definition}
  An  $m\times n$ matrix $\mat{A}$  is left triangular if $\mat{A}_{i,j}=0$ for all $i> n-j$.
\end{definition}

The left triangular part of a matrix $\mat{A}$, denoted by $\LTP(\mat{A})$ will
refer to the left triangular matrix extracted from it.
We will need the following property on the left triangular part of the product
of a matrix by a triangular matrix.

\begin{lemma}\label{lem:leftproductup}
  Let  $\mat{A}=\mat{B}\mat{U}$ be an $m\times n$ matrix  where 
$\mat{U}$ is $n\times n$ upper triangular. Then
  $\LTP(\mat{A}) = \LTP(\LTP(\mat{B}) \mat{U})$.
\end{lemma}

\begin{proof}
  Let $\mat{\bar A} = \LTP(\mat{A}), \mat{\bar B} =\LTP(\mat{B})$.
For $j\leq n-i,$ we have
$
\mat{\bar A}_{i,j} = \sum_{k=1}^n \mat{B}_{i,k} \cdot \mat{U}_{k,j} = \sum_{k=1}^j \mat{B}_{i,k} \cdot \mat{U}_{k,j}
$
as $\mat{U}$ is upper triangular. Now for $k\leq j \leq n-i$, $\mat{B}_{i,k} =
\mat{\bar B}_{i,k}$,
which
proves that the left triangular part of $\mat{A}$ is that of $\LTP(\mat{B})\mat{U}$.
\end{proof}

Applying Lemma~\ref{lem:leftproductup} on $\mat{A}^T$ yields Lemma~\ref{lem:leftproductlow}
\begin{lemma}\label{lem:leftproductlow}
  Let $\mat{A}=\mat{L}\mat{B}$  be an $m\times n$ matrix where
$\mat{L}$ is $m\times m$ lower triangular. Then
  $\LTP(\mat{A}) = \LTP(\mat{L}\LTP(\mat{B}))$.
\end{lemma}

Lastly, we will extend the notion of order of quasiseparability to left triangular
matrices, in the natural way: the order of left quasiseparability is the maximal rank of any leading $k\times
(n-k)$ sub-matrix. When no confusion may occur, we will abuse the definition and
simply call it the order of quasiseparability.

\subsection{PLUQ decomposition}
We recall that for any $m\times n$ matrix $\mat{A}$ of rank $r$, there exist a
PLUQ decomposition $\mat{A}=\mat{P}\mat{L}\mat{U}\mat{Q}$ where $\mat{P}$ is an $m\times m$ permutation
matrix, $\mat{Q}$ is an $n\times n$ permutation matrix, $\mat{L}$ is an $m\times r$ unit lower
triangular matrix, and $\mat{U}$ is an $r\times n$ upper triangular matrix.
matrix. It is not unique, but once the permutation matrices $\mat{P}$ and
$\mat{Q}$ are fixed, the triangular factors $\mat{L}$ and $\mat{U}$ are unique,
since the matrix $\mat{P}^T\mat{A}\mat{Q}^T$ has generic rank profile and
therefore has a unique LU decomposition.

\subsection{The rank profile matrix}
We will use a matrix invariant, introduced in~\cite[Theorem~1]{DPS15}, that summarizes the
information on the ranks of any leading sub-matrices of a given input matrix.

\begin{definition}{\citep[Theorem~1]{DPS15}}\label{def:rpm}
  The rank profile matrix of an $m\times n$ matrix $\mat{A}$ of rank $r$ is the unique $m \times
  n$ matrix $\mathcal{R}_\mat{A}$, with only $r$ non-zero coefficients, all equal to one, located on
  distinct rows and columns such that any leading sub-matrices of $\mathcal{R}_\mat{A}$ has
  the same rank as the corresponding leading sub-matrix in $\mat{A}$.
\end{definition}

This invariant can be computed in just one Gaussian elimination of the matrix
$\mat{A}$, at the cost of $\GO{mnr^{\omega-2}}$ field operations~\citep{DPS15}, provided some conditions on the pivoting strategy being used. It is obtained
from the corresponding PLUQ decomposition as the product 
$$\RPM_\mat{A} = \mat{P}\begin{bmatrix}  \mat{I}_r\\&\mat{0}_{(m-r)\times (n-r)}\end{bmatrix}\mat{Q}.$$
%

\begin{remark}
  The notion of rank profile matrix orginates from Bruhat's matrix decomposition
  for non-singular matrices in~\cite{Bru56}, where uniqueness of this permutation was
  established. It has then been generalized in~\cite{Tyr97} and~\cite{MH07} for all matrices and
  in~\cite{Mal10} with the LEU decomposition, where it appears as the $E$
  factor.
  The connection to rank profiles introduced in~\cite{DPS15} was the motivation
  for its name.
\end{remark}

We also recall in Theorem~\ref{th:PLPT} an important property of such PLUQ
decompositions revealing the rank profile matrix.

\begin{property} [{\citep[Th.~24]{DPS16}, \citep[Th.~1]{DPS13}}] \label{th:PLPT}
Let $\mat{A}=\mat{P}\mat{L}\mat{U}\mat{Q}$, a PLUQ decomposition revealing the
rank profile matrix of $\mat{A}$.
Then, 
$\mat{P}\begin{bmatrix}  \mat{L}&\mat{0}_{m\times (m-r)}\end{bmatrix}\mat{P}^T$ 
is lower triangular and 
$\mat{Q}^T\begin{bmatrix}  \mat{U}\\\mat{0}_{(n-r)\times n}\end{bmatrix}\mat{Q}$ 
is upper triangular.
\end{property}

\begin{lemma}\label{lem:preserverpm} The rank profile matrix invariant is preserved by multiplication
  \begin{enumerate}
  \item to the left  with an invertible lower triangular matrix,
  \item to the right with an invertible upper triangular matrix.
  \end{enumerate}
\end{lemma}
\begin{proof}
  Let $\mat{B}=\mat{L}\mat{A}$ for an invertible lower triangular matrix $\mat{L}$. Then for
  any $i\leq m, j\leq n$, $\text{rank}(\mat{B}_{1..i,1..j})
  = \text{rank}(\mat{L}_{1..i,1..i} \mat{A}_{1..i,1..j}) = \text{rank}(\mat{A}_{1..i,1..j})$ . Hence $\RPM_\mat{B} = \RPM_\mat{A}$.
\end{proof}
\section{Computing the orders of quasiseparability}\label{sec:comporders}

Let $\mat{M}$ be an $n\times n$ matrix of which one wants to determine the
quasiseparable orders $(r_L,r_U)$. Let $\mat{L}$ and $\mat{U}$ be respectively the
lower triangular part and the upper triangular part of $\mat{M}$. 

Multiplying on the left by $\mat{J}_n$, the unit anti-diagonal matrix, inverses the row order while multiplying on the right by $\mat{J}_n$ inverses
the column order.
Hence both $\mat{J}_n \mat{L}$ and
$\mat{U}\mat{J}_n$ are left triangular matrices. Remark that
conditions~\eqref{eq:quasisep:low} and~\eqref{eq:quasisep:up} state that all 
leading $k\times (n-k)$ sub-matrices of $\mat{J}_n \mat{L}$ and
$\mat{U}\mat{J}_n$ have rank no greater than $r_L$ and $r_U$ respectively.
We will then use the rank profile matrix of these two left triangular matrices
to find these parameters.

\subsection{From a rank profile matrix}

First, note that the rank profile matrix of a left triangular matrix is not necessarily
left triangular. For example, the rank profile matrix of $
\begin{smatrix}
  1 & 1 & 0\\
  1 & 0 & 0\\
  0 & 0 & 0
\end{smatrix}
$
is
$
\begin{smatrix}
  1 &0&0\\
  0 &1&0\\
  0 &0 &0
\end{smatrix}
$.
However, only the left triangular part of the rank profile matrix is sufficient to
compute the left quasiseparable orders. 

Suppose for the moment that the left-triangular part of the rank profile matrix
of a left triangular matrix is given (returned by a function LT-RPM). It remains to enumerate all leading $k\times (n-k)$
sub-matrices and find the one with the largest number of non-zero elements.
Algorithm~\ref{alg:qsrank} shows how to compute the largest rank of all leading
sub-matrices of such a matrix. Run on $\mat{J}_n\mat{L}$ and $\mat{U}\mat{J}_n$, it returns
successively the quasiseparable orders $r_L$ and $r_U$.
\begin{algorithm}[ht]
{
  \caption{QS-order} \label{alg:qsrank}
  \begin{algorithmic}
    \Require{$\mat{A}$, an $n\times n$ matrix}
    \Ensure{$\max\{\text{rank}(\mat{A}_{1..k,1..n-k}) : 1\leq k\leq n-1\}$}
    \State $\mathcal{R} \leftarrow \text{LT-RPM}(\mat{A})$\Comment{The left
      triangular part of the rank profile matrix of $\mat{A}$}
      \State rows $\leftarrow $ (False,\dots,False)
      \State cols $\leftarrow $ (False,\dots,False)
      \ForAll{$(i,j)$ such that $\mathcal{R}_{i,j}=1$}
        \State rows[i] $\leftarrow$ True
        \State cols[j] $\leftarrow$ True
      \EndFor
      \State $s,r \leftarrow 0$
      \For{$i=1\dots n-1$}
        \State \textbf{if} rows[$i$] \textbf{then} $r\leftarrow r+1$
        \State \textbf{if} cols[$n-i+1$] \textbf{then} $r\leftarrow r-1$
        \State $s\leftarrow \max(s,r)$
      \EndFor
      \State \textbf{return} $s$
  \end{algorithmic}
}
\end{algorithm}
  
This algorithm runs in $\GO{n}$ provided that the rank profile matrix $\mat{R}$
is stored in a compact way, e.g. using a vector of $r$ pairs of pivot indices ($[(i_1,j_1),\dots,(i_r,j_r)]$.

\subsection{Computing the rank profile matrix of a left triangular matrix}

We now deal with the missing component: computing the left triangular part of
the rank profile matrix of a left triangular matrix.

\subsubsection{From a PLUQ decomposition}

A first approach is to run any Gaussian elimination algorithm that can reveal
the rank profile matrix, as described in~\cite{DPS15}. In
particular, the PLUQ decomposition algorithm of~\cite{DPS13} computes the rank
profile matrix of $\mat{A}$ in $\GO{n^2r^{\omega-2}}$ where
$r=\text{rank}(\mat{A})$. 
However this estimate may be pessimistic as it does not take into account the left
triangular shape of the matrix. 
Moreover, it does not depend on the left quasiseparable order $s$
but on the rank $r$, which could be much higher.

\begin{remark}\label{rem:tradeoff}
The discrepancy between the rank $r$ of a left triangular matrix and its
quasiseparable order arises from the location of the pivots in its rank profile
matrix.
Pivots located near the top left corner of the matrix  are shared by
many leading sub-matrices, and are therefore likely to contribute to the
quasiseparable order. On the other hand, pivots near the main anti-diagonal can
be numerous, but do not add up to a large quasiseparable order.
As an illustration, consider the two following extreme cases:
\begin{enumerate}
\item a matrix $\mat{A}$ with generic rank profile. Then the leading
  $r\times r$ sub-matrix of $\mat{A}$ has rank $r$ and the 
  quasiseparable order is $s=r$.
\item  the matrix with $n-1$ ones immediately above the main anti-diagonal. It has rank
  $r=n-1$ but quasiseparable order $1$.
\end{enumerate}
\end{remark}

Remark~\ref{rem:tradeoff} indicates that in the unlucky cases when $r\gg s$,
the computation should reduce to instances of smaller sizes, hence a trade-off
should exist between, on one hand, the discrepancy between $r$ and $s$, and
on the other hand, the dimension $n$ of the problems. All contributions
presented in the remaining of the paper are based on such trade-offs.

\subsubsection{A dedicated algorithm}

In order to reach a complexity depending on $s$ and not $r$, we adapt in
Algorithm~\ref{alg:LTElim}  the tile recursive 
algorithm of~\cite{DPS13}, so that the left triangular structure of the input
matrix is preserved and can be used to reduce the amount of computation.
In this algorithm, the input matrix is modified in-place, and comments keep
track of its current value. In particular, the upper and lower
triangular factors obtained after a PLUQ decomposition are stored one above the
other on the same storage, which is represented by the notation
$\begin{bmatrix}  \mat{L}\backslash \mat{U}\end{bmatrix}$.

Algorithm~\ref{alg:LTElim} does not assume that the input matrix is
left triangular, as it will be called recursively with arbitrary matrices, but
guarantees to return the left triangular part of the rank profile matrix. 
\begin{algorithm}[ht]
{
  \begin{algorithmic}[1]
    \caption{LT-RPM: Left Triangular part of the Rank Profile Matrix} \label{alg:LTElim}
    \Require{$\mat{A}$: an $n\times n$ matrix}
    \Ensure{$\mathcal{R}$: the left triangular part of the RPM of $\mat{A}$}
    \State \textbf{if} $n=1$ \textbf{then} \textbf{return} $[0]$
    \State {Split $\mat{A} = \begin{bmatrix} \mat{A_{1}} & \mat{A_{2}}\\\mat{A_{3}} \end{bmatrix}$ 
      where $\mat{A_{3}}$ is $\lfloor \frac{n}{2} \rfloor \times \lfloor
      \frac{n}{2}\rfloor$}
    \State Compute a PLUQ decomposition $\mat{A_{1}} =
    \mat{P_1} \begin{bmatrix}\mat{L_1}\\\mat{M_1}\end{bmatrix}\begin{bmatrix}\mat{U_1}&\mat{V_1}\end{bmatrix}
    \mat{Q_1}$  revealing the RPM 
    \State $\mathcal{R}_1  \leftarrow \mat{P_1} \begin{bmatrix}
      \mat{I_{r_1}}\\&\mat{0}\end{bmatrix}\mat{Q_1}$ where $r_1=\text{rank}(\mat{A_1})$.
    \State $\begin{bmatrix} \mat{B_1}\\ \mat{B_2}\end{bmatrix}
    \leftarrow \mat{P_1}^T\mat{A_{2}}$ 
    \State $
    \begin{bmatrix}
      \mat{C_1}&\mat{C_2}
    \end{bmatrix}
    \leftarrow \mat{A_{3}}\mat{Q_1}^T$ 
    \Comment Here $\mat{A} =
    \left[\begin{array}{cc|c}
        \mat{L_1} \backslash \mat{U_1}& \mat{V_1}& \mat{B_1}\\
        \mat{M_1}               & \mat{0}  & \mat{B_2}\\
        \hline
        \mat{C_1}               & \mat{C_2}& \\
      \end{array}\right]$.
    \State $\mat{D}\leftarrow \mat{L_1}^{-1}\mat{B_1}$ 
    \State $\mat{E}\leftarrow \mat{C_1}\mat{U_1}^{-1}$ 
    \State $\mat{F}\leftarrow \mat{B_2}-\mat{M_1}\mat{D}$ 
    \State $\mat{G}\leftarrow \mat{C_2}-\mat{E}\mat{V_1}$
    \Comment Here $\mat{A}=
    \left[\begin{array}{cc|c}
        \mat{L_1} \backslash \mat{U_1}& \mat{V_1}& \mat{D}\\
        \mat{M_1}               & \mat{0}  & \mat{F}\\
        \hline
        \mat{E}               & \mat{G}& \\
      \end{array}\right]$.
    
    \State $\mat{H} \leftarrow \mat{P_1}  \begin{bmatrix} \mat{0}_{r_1 \times \frac{n}{2}} \\ \mat{F} \end{bmatrix}$
    \State $\mat{I} \leftarrow  \begin{bmatrix} \mat{0}_{r_1 \times \frac{n}{2}} & \mat{G} \end{bmatrix} \mat{Q_1}$
    \State $\mathcal{R}_2 \leftarrow \texttt{LT-RPM}(\mat{H})$
    \State $\mathcal{R}_3 \leftarrow \texttt{LT-RPM}(\mat{I})$
    \State \textbf{return} $\mathcal{R} \leftarrow  \begin{bmatrix}     \mathcal{R}_1 & \mathcal{R}_2\\ \mathcal{R}_3   \end{bmatrix}$
  \end{algorithmic}
}
\end{algorithm}
While the top left quadrant $\mat{A}_1$ is eliminated using any PLUQ decomposition algorithm
revealing the rank profile matrix, the top right and bottom left quadrants
are handled recursively.

\begin{theorem}\label{th:LTRPM}
  Given an $n\times n$ input matrix $\mat{A}$ with left quasiseparable order $s$,
  Algorithm~\ref{alg:LTElim} computes the left triangular part of the rank 
  profile matrix of $\mat{A}$ in $\GO{n^2s^{\omega-2}}$ field operations.
\end{theorem}

\begin{proof}
First remark that 
$$
\mat{P_1} \begin{bmatrix}  \mat{D}\\\mat{F}\end{bmatrix} = \underbrace{\mat{P_1}
\begin{bmatrix} \mat{L_1}^{-1}\\ -\mat{M_1}\mat{L_1}^{-1} & \mat{I_{n-r_1}} \end{bmatrix} \mat{P_1}^T}_{\mat{L}} \mat{P_1}
\begin{bmatrix}   \mat{B_1}\\\mat{B_2} \end{bmatrix} = \mat{L}\mat{A}_2.
$$
Hence
$$
\mat{L} \begin{bmatrix}  \mat{A_1} & \mat{A_2}\end{bmatrix} =
\mat{P_1}
\left[\begin{array}{c|c}
  \begin{bmatrix}\mat{U_1} & \mat{V_1}\end{bmatrix}\mat{Q_1} & \mat{D} \\
   \mat{0}   & \mat{F}
\end{array}
\right].
$$
From Theorem~\ref{th:PLPT}, the matrix $\mat{L}$ is lower triangular and by Lemma~\ref{lem:preserverpm}
the rank profile matrix of $\begin{bmatrix}  \mat{A_1} & \mat{A_2}\end{bmatrix}$ equals  that of $\mat{P_1}
\left[\begin{array}{c|c}
  \begin{bmatrix}\mat{U_1} & \mat{V_1}\end{bmatrix}\mat{Q_1} & \mat{D} \\
   \mat{0}   & \mat{F}
\end{array}
\right]$.
Now as $\mat{U_1}$ is upper triangular and non-singular, this rank profile matrix is
in turn that of
$\mat{P_1}
\left[\begin{array}{c|c}
  \begin{bmatrix}\mat{U_1} & \mat{V_1}\end{bmatrix}\mat{Q_1} & \mat{0} \\
   \mat{0}   & \mat{F}
\end{array}
\right]$ and its left triangular part is $\begin{bmatrix}  \mathcal{R}_1 & \mathcal{R}_2\end{bmatrix}$.

By a similar reasoning, $\begin{bmatrix}  \mathcal{R}_1& \mathcal{R}_3\end{bmatrix}^T$
is the left triangular part of the rank profile matrix of $\begin{bmatrix}  \mat{A_1}& \mat{A_3}\end{bmatrix}^T$, which shows that the algorithm is correct.

Let $s_1$ be the left quasiseparable order of $\mat{H}$ and $s_2$ that of $\mat{I}$.
The number of field operations required to run Algorithm~\ref{alg:LTElim} is
$$
T(n,s)=\alpha
r_1^{\omega-2}n^2+T_{\text{LT-RPM}}(n/2,s_1)+T_{\text{LT-RPM}}(n/2,s_2)
$$
 for a
positive constant $\alpha$.
We will prove by induction that $T(n,s)\leq 2\alpha s^{\omega-2}n^2$.

Again, since $\mat{L}$ is lower triangular, the rank profile matrix of
$\mat{L}\mat{A_2}$ is that of $\mat{A}_2$ and the quasiseparable orders of
the two matrices are the same. Now $\mat{H}$ is the matrix $\mat{L}\mat{A_2}$
with some rows zeroed out, hence $s_1$, the
quasiseparable order of $\mat{H}$ is no greater than that of $\mat{A_2}$ which 
is less or equal to $s$.
Hence $\max(r_1,s_1,s_2) \leq s $ and we obtain $T(n,s) \leq  \alpha
s^{\omega-2}n^2 + 4 \alpha s^{\omega-2}(n/2)^2 = 2\alpha s^{\omega-2}n^2$.
\end{proof}

\section{New structured representations for quasiseparable matrices}
\label{sec:generators}

In order to introduce fast matrix arithmetic in the algorithms computing with
quasiseparable matrices, we introduce in this section three new structured
representations: the Recursive Rank Revealing (RRR) representation, the Bruhat
representation, and finally its compact version, the Compact Bruhat (CB) representation.

%

\subsection{The Recursive Rank Revealing representation}

This a simplified version of the HSS representation. It uses in the same manner
a recursive splitting of the matrix in a quad-tree, and each off-diagonal block at
each recursive level is represented by a rank revealing factorization.

\begin{definition}
  [RR: Rank revealing factorization]
  A rank revealing factorization (RR) of an $m\times n$ matrix $\mat{A}$
    of rank $r$ is a pair of matrices $\mat{L}$ and $\mat{R}$ of dimensions $m\times r$ and
$r\times n$ respectively, such that $\mat{A}=\mat{L}\mat{R}$.
\end{definition}
For instance, a PLUQ decomposition is a rank revealing factorization. One can
either store explicitely the two factors  $\mat{P}\mat{L}$ and $\mat{U}\mat{Q}$
or only consider the factors $\mat{L}$ and $\mat{U}$ keeping in mind that
permutations need to be applied on the left and on the right of the product.

\begin{definition}[RRR: Recursive Rank Revealing representation]
  A recursive rank revealing  (RRR) representation of an $n\times n$ quasiseparable matrix 
$ \mat{A}=\begin{bmatrix}\mat{A}_{11}&\mat{A}_{12}\\ \mat{A}_{21}&\mat{A}_{22}\end{bmatrix}$ 
of order $s$ is formed by a rank revealing factorization of $\mat{A}_{12}$ and
$\mat{A}_{21}$ and applies recursively for the representation of $\mat{A}_{11}$ and
$\mat{A}_{22}$.
\end{definition}

A Recursive Rank Revealing representation forms a binary tree where each node
correspond to a diagonal block of the input matrix, and contains the Rank
Revealing factorization of its off-diagonal quadrants.

If $\mat{A}$ is $(r_L,r_U)$-quasiseparable, then all off-diagonal blocks in its
lower part have rank bounded by $r_L$, and their rank revealing factorizations
take advantage of this low rank until a block dimension $n/2^k \approx r_L$
where a dense representation is used. The same applies for the upper triangular part
with quasiseparable order $r_U$.
This representation uses $O(sn\log \frac{n}{s})$ space where $s=\max(r_L, r_U)$.

\subsection{The Bruhat representation}

This structured representation is closely related to the notion of the rank
profile matrix and the LEU decomposition of~\cite{Mal10}.
Contrarily to the RRR or the HSS representations, it is not depending on a specific recursive cutting of the matrix.
For this representation, and its compact version that will be studied in section~\ref{sec:CB:def}, the
lower and the upper triangular parts are represented 
independently. We will therefore  treat them in a unified way, showing how to
represent a left triangular matrix. Recall that if $\mat{L}$ is lower triangular
and $\mat{U}$ is upper triangular then both $\mat{J_n}\mat{L}$ and $\mat{U}\mat{J_n}$ are left triangular.

Given a left triangular matrix $\mat{A}$ of quasiseparable order $s$ and a PLUQ decomposition of it,
revealing its rank profile matrix $\mat{R}$, the Bruhat generator consists in the three matrices 
\begin{eqnarray}
  \label{eq:storage}
 \mathcal{L} &=&\LTP(\mat{P}\begin{bmatrix}\mat{L}&\mat{0}\end{bmatrix}\mat{Q}) \label{eq:storage:L},\\
 \mathcal{R} &=& \LTP(\mat{R}),\\
 \mathcal{U} &=& \LTP(\mat{P}\begin{bmatrix}    \mat{U}\\ \mat{0}  \end{bmatrix} \mat{Q})\label{eq:storage:U}.
\end{eqnarray}

Lemma~\ref{lem:LEUstorage} shows that these three matrices suffice to recover
the initial left triangular matrix.
\begin{lemma}\label{lem:LEUstorage}
$\mat{A}= \LTP( \mathcal{L}\mathcal{R}^T \mathcal{U})$
\end{lemma}

\begin{proof}
 $   \mat{A} =  \mat{P}\begin{bmatrix}  \mat{L} & \mat{0}_{m\times (n-r)}\end{bmatrix} \mat{Q}  \mat{Q}^T \begin{bmatrix}  \mat{U}\\
  \mat{0}_{(n-r)\times n}\end{bmatrix} \mat{Q}.
$
From Theorem~\ref{th:PLPT}, the matrix $\mat{Q}^T\begin{bmatrix}  \mat{U}\\  \mat{0}\end{bmatrix} \mat{Q}$
is upper triangular and the matrix $\mat{P}
\begin{bmatrix}  \mat{L}&\mat{0}\end{bmatrix}\mat{P}^T$ is lower triangular.
Applying Lemma~\ref{lem:leftproductup} yields
$
\mat{A}=\LTP(\mat{A}) = \LTP(\mathcal{L}\mat{Q}^T\begin{bmatrix}  \mat{U}\\
  \mat{0}\end{bmatrix} \mat{Q} ) = \LTP(\mathcal{L}\mat{R}^T \mat{P}\begin{bmatrix}  \mat{U}\\  \mat{0}\end{bmatrix} \mat{Q} ),
$
where $\mat{R}=\mat{P}\begin{smatrix}  \mat{I}_r\\&0\end{smatrix}\mat{Q}$.
Then, as $\mathcal{L}\mat{R}^T$ is the matrix $\mat{P}
\begin{bmatrix}  \mat{L}&\mat{0}\end{bmatrix}\mat{P}^T$ with some coefficients
zeroed out, it is lower triangular, hence applying again
Lemma~\ref{lem:leftproductlow} yields
\begin{equation}\label{eq:ALEU}
\mat{A}= \LTP(\mathcal{L}\mat{R}^T \mathcal{U}).
\end{equation}
Consider any non-zero coefficient $e_{j,i}$ of $\mat{R}^T$ that is not in its the left
triangular part, i.e. $j>n-i$. Its contribution to the product
$\mathcal{L}\mat{R}^T$, is only of the form $\mathcal{L}_{k,j}e_{j,i}$. However
the leading coefficient in column $j$ of $\mat{P}\begin{bmatrix}
  \mat{L}&\mat{0}\end{bmatrix}\mat{Q}$ is precisely at position $(i,j)$. Since
$i>n-j$, this means that the $j$-th column of $\mathcal{L}$ is all zero, and
therefore $e_{i,j}$ has no contribution to the product.
Hence we finally have
 $\mat{A} = \LTP(\mathcal{L}\mathcal{R}^T\mathcal{U})$.
\end{proof}


We now analyze the space required by this generator.
\begin{lemma}\label{lem:size}
  Consider an $n\times n$ left triangular rank profile matrix $\mat{R}$ with quasiseparable order $s$. Then a
  left triangular matrix $\mat{L}$ all zero except at the positions of the pivots of
  $\mat{R}$ and below these pivots, does not contain more than $s(n-s)$ non-zero
  coefficients. 
\end{lemma}

\begin{proof}
  Let $p(k)=\text{rank}(\mat{R}_{1..k,1..n-k})$. The value $p(k)$ indicates the
  number of non zero columns located in the $k\times n-k$ leading sub-matrix of
  $\mat{L}$ and is therefore an upper bound on the number of non-zero elements
  in row $k$ of $\mat{L}$. Consequently the sum $\sum_{k=1}^{n-1} p(k)$ is an upper bound on the
  number of non-zero coefficients in $\mat{L}$.
  Since $p(k)\leq s$, it is bounded by $sn$. More precisely, there is no more
  than $k$ pivots in the first $k$ columns and the first $k$ rows, hence
  $p(k)\leq k$  and $p(n-k)\leq k$ for $k\leq s$. The bound becomes $s(s+1)
  +(n-2s-1)s = s(n-s)$.
\end{proof}

\begin{corollary}
The Bruhat generator $\left(\mathcal{L},\mathcal{R},\mathcal{U}\right)$
uses $2s(n-s)$ field coefficients and $\GO{n}$ additional indices to represent a
left triangular matrix.
\end{corollary}

\begin{proof}
The leading column elements of  $\mathcal{L}$ are located at
the pivot positions of the left triangular rank profile matrix $\mathcal{R}$.
Lemma~\ref{lem:size} can therefore be applied to show that this matrix occupies
no more than $s(n-s)$ non-zero coefficients.
The same argument applies to the matrix $\mathcal{U}$.
\end{proof}

Figure~\ref{fig:bruhatstorage} illustrates this generator on  a left triangular
matrix of quasiseparable order~$5$. 
\begin{figure}[htbp]
\begin{center}
   \includegraphics[height=\myfigsize]{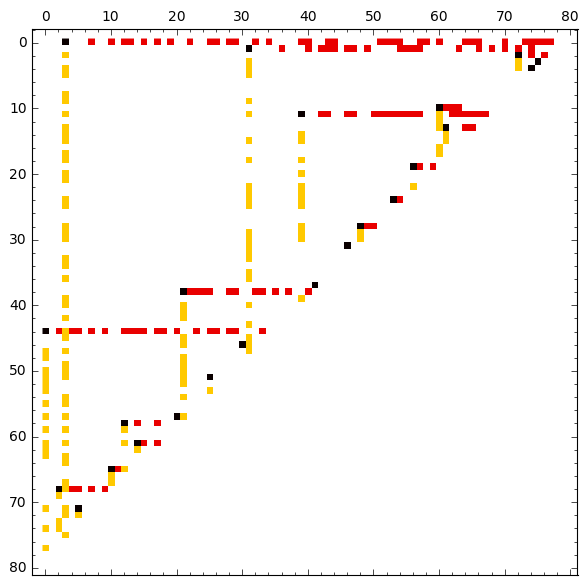}\\
\end{center}
\caption{Support of the $\mathcal{L}$ (yellow), $\mathcal{R}$ (black) and
  $\mathcal{U}$ (red) matrices of the Bruhat generator for a $80\times 80$ left triangular
  matrix of quasiseparable order $5$.} \label{fig:bruhatstorage}
\end{figure}
As the supports of $\mathcal{L}$ and
$\mathcal{U}$ are disjoint, the two matrices can be shown on the same
left triangular matrix. The pivots of $\mathcal{R}$ (black) are the leading
coefficients of every non-zero row of $\mathcal{U}$ and non-zero column of $\mathcal{L}$.

\begin{corollary}
  Any $(r_L,r_U)$-quasiseparable matrix of dimension $n \times n$ can be
  represented by a generator using no more than $2n(r_L+r_U)+n-2(r_L^2 -2r_U^2)$ field elements and $2(r_L+r_U)$ indices.
\end{corollary}
\begin{proof}
  This estimate is obtained as the space required for the Bruhat
  representation of the upper and lower triangular parts of the matrix, with 
   $n$ coefficients for the main diagonal. The $2(r_L+r_U)$ indices correspond
  to the storge of the pivot positions of the two rank profile matrices.
\end{proof}

\subsection{The compact Bruhat representation}
\label{sec:CB:def}

The scattered structure of the Bruhat generator makes it not amenable to the use
of fast matrix arithmetic. 
We therefore propose here a compact variation on it, called the compact Bruhat,
that will be used to derive algorithms taking advantage of fast matrix
multiplication.
This structured representation relies on the 
generalized Bruhat decomposition described in~\cite{MH07}, thanks to the
connection with the rank profile matrix made in~\cite{DPS16}.

\begin{theorem}[Generalized Bruhat decomposition~\citep{MH07,DPS16}] \label{th:genbruhat}
  For any $m\times n$ matrix $\mat{A}$ of rank $r$, there exist an $m\times r$
  matrix $\mat{C}$ in column echelon form, an $r\times n$ matrix $\mat{E}$ in row echelon
  form, and an $r\times r$ permutation matrix $\mat{R}$ such that
  $\mat{A}=\mat{C}\mat{R}\mat{E}$.
\end{theorem}

We will also need an additional structure on the echelon form factors.
\begin{definition} 
  Two non-zero columns of matrix are non-overlapping if one has its leading element
  below the trailing element of the other.
\end{definition}
\begin{definition}
A matrix is $s$-overlapping if any sub-set of $s+1$ of its non-zero columns contains at least a pair that are non-overlapping.
    \end{definition}

The motivation for introducing this structure is that left triangular matrices of
quasiseparable order $s$ have a generalized Bruhat decomposition with echelon form factors
$\mat{C}$ and $\mat{E}$ that are  $s$-overlapping.

\begin{theorem}\label{th:genbruhatoverlap}
  For any $n\times n$ left triangular matrix $\mat{A}$ of quasiseparable
  order $s$ and of rank $r$, there is a generalized Bruhat decomposition of the form
    $\mat{A}=\LTP(\mat{C}\mat{R}\mat{E})$ where $\mat{C}$ and $\mat{E}^T$ are $s$-overlapping.
\end{theorem}

\begin{proof}
Let $(\mathcal{L},\mathcal{R},\mathcal{U})$ be a Bruhat generator for
$\mat{A}$. The matrix $\mathcal{L}$ is $s$-overlapping: otherwise, there would be
a subset $S$ of $s+1$ of columns such that no pair of them is
non-overlapping. Let $((i_1,j_1),\dots,(i_{s+1},j_{s+1}))$ be the coordinates of
their leading elements sorted by increasing row index : $i_1<i_2<\dots<i_{s+1}$.
Since $\mathcal{L}$ is left triangular, $j_{s+1}\leq n-i_{s+1}$. The
trailing elements of every other column of $S$ must be below row $i_{s+1}$,
hence, $j_k\leq n-i_{s+1}$ for all $k\leq s$ since $\mathcal{L}$ is left
triangular. Consequently the $i_{s+1}\times (n-i_{s+1})$ leading submatrix of
$\mathcal{L}$ contains $s+1$ pivots, a contradiction.
The same reasonning
applies to show that $\mat{E}^T$ is $s$-overlapping.
Consider the permutation matrix $\mathcal{Q}$ such that
$\mathcal{L}\mathcal{Q}= \begin{bmatrix}  \mat{C}&\mat{0}_{m\times (n-r)}\end{bmatrix}$ is in column
echelon form.
Similarly let $\mathcal{P}$ be the permutation matrix such that
$\mathcal{P}\mathcal{U}= \begin{bmatrix}\mat{E}\\\mat{0}_{(m-r)\times n}\end{bmatrix}$, and
remark that $\mat{R} = \begin{bmatrix}  \mat{I}_r&\mat{0}\end{bmatrix} \mathcal{Q}^T\mathcal{R}^T\mathcal{P}^T\begin{bmatrix}\mat{I}_r\\\mat{0}\end{bmatrix}$ is a permutation matrix and verifies $\mat{A}=\LTP(\mat{C}\mat{R}\mat{E})$.
\end{proof}

The $s$-overlapping shape of the echelon form factors in the generalized Bruhat
decomposition allow to further compress it as follows.

\begin{proposition}\label{prop:soverlap}
Any $s$-overlapping $m\times r$  matrix $\mat{A}$  can be written
$\mat{A}=(\mat{D} + \mat{S}\mat{T})\mat{P}$ where $\mat{P}$ is a permutation
matrix, $\mat{T}\in\{0,1\}^{r\times r}$ has at most one non zero element per row and
$\mat{D}=\text{Diag}(\mat{D}_1,\dots,\mat{D}_{t})$,
$\mat{S}=
\begin{smatrix}
  \mat{0}\\
  \mat{S}_2 & \mat{0}\\
   & \ddots &\ddots\\
   & &\mat{S}_t
\end{smatrix}$  where each $\mat{D}_i$ and $\mat{S}_i$ is $k_i\times s$, except
$\mat{D}_{t}$ and $\mat{S}_t$ having possibly fewer columns than~$s$  and
$\sum_{i=1}^tk_i=n$.
\end{proposition}

Intuitively, the permutation $\mat{P}^T$ sorts the columns of $\mat{A}$ in
increasing order of their leading row index. Cutting the columns in slices
of dimension $s$ makes $\mat{A}\mat{P}^T$ block lower triangular. The block
diagonal is $\mat{D}$, and the remaining part can be folded into a block
sub-diagonal matrix $\mat{S}$ thanks to the $s$-overlapping property.

\begin{algorithm}[htbp]
{
\caption{Compress-to-Block-Bidiagonal} \label{alg:compress}
 \begin{algorithmic}[1]
\Require{$\mat{A}$: an $s$-overlapping matrix}
\Ensure{$\mat{D}, \mat{S}, \mat{T}, \mat{P}$: such that $\mat{A}=(\mat{D} +
  \mat{S}\mat{T})\mat{P}$ as in
  Proposition~\ref{prop:soverlap}.}

\State $\mat{P}\leftarrow$ a permutation sorting the columns of $\mat{A}$ by
increasing row position of their leading coefficient.
\State $\mat{C}\leftarrow \mat{A}  \mat{P} \begin{smatrix}  \mat{I}_r\\
  \mat{0}\end{smatrix}$ where $r$ is the number of non-zero columns in $\mat{A}$ 
\State Split $\mat{C}$ in column slices of width $s$.  \Comment $\mat{C}=
\begin{smatrix}
  \mat{C}_{11}& \\
  \mat{C}_{21} & \mat{C}_{22} \\
  \vdots & \vdots& \ddots\\
  \mat{C}_{t1} & \mat{C}_{t2} & \dots & \mat{C}_{tt}
\end{smatrix}$ where $\mat{C}_{ii}$ is $k_i\times s$ $\forall i<t$.
\State $\mat{D} \leftarrow \text{Diag}(\mat{C}_{11},\dots,\mat{C}_{tt})$
\State \label{step:CminusD} $\mat{C} \leftarrow \mat{C}-\mat{D} = 
\begin{smatrix}\mat{0}\\ \mat{C}_{21}  \\
\vdots &  \ddots&\ddots\\
  \mat{C}_{t1} &  \dots & \mat{C}_{t,t-1}& \mat{0}
\end{smatrix}
$
\State $\mat{T}\leftarrow \mat{I}_n$
\For{$i=3\dots t$} \label{step:loop}
   \For{each non zero column $j$ of $ \begin{smatrix}  \mat{C}_{i,i-2} \\ \dots \\
       \mat{C}_{t,i-2}\end{smatrix} $}
      \State \label{step:zerocol} Let $k$ be a zero column of 
$ \begin{smatrix}  \mat{C}_{i,i-1} \\ \dots \\ \mat{C}_{t,i-1}\end{smatrix}$
      \State \label{step:movecol} Move column $j$ in $\begin{smatrix}  \mat{C}_{i,i-2}\\ \vdots\\\mat{C}_{t,i-2}   \end{smatrix}$ to 
   column $k$ in  $\begin{smatrix}  \mat{C}_{i,i-1}\\ \vdots\\\mat{C}_{t,i-1}   \end{smatrix}$.
      \State $\mat{T}\leftarrow (\mat{I}_n + \mat{\Delta}^{(k,j)}-\mat{\Delta}^{(k,k)})\times \mat{T}$
   \EndFor
\EndFor

\State $\mat{S}\leftarrow\mat{C}=\begin{smatrix}
  \mat{0}& \\
  \mat{C}_{21} & \mat{0} \\
       & \ddots& \ddots\\
       &       & \mat{C}_{t,t-1} & \mat{0}
\end{smatrix}$
\State \Return $(\mat{D},\mat{S}, \mat{T}, \mat{P}
)$
\end{algorithmic}
}
\end{algorithm}
Algorithm~\ref{alg:compress} is a constructive proof of
Proposition~\ref{prop:soverlap}, computing a compact representation of any
$s$-overlapping matrix.
\begin{proof}
Since $\mat{A}$ is $s$-overlapping, there exists a permutation $\mat{P}$ such that
$\mat{C}=\mat{A}\mat{P}$ is block lower triangular, with blocks of column
dimension $s$ except possibly the last one of column dimension $\leq s$.
Note that for every $i$, the dimensions of the blocks $\mat{S}_i$ and $\mat{D}_i$ are that of the
 block $\mat{C}_{ii}$: $k_i\times s$. 
We then prove that there always exists a zero column to pick at
step~\ref{step:zerocol}.
In the first row of $\begin{bmatrix} \mat{C}_{i1}&\dots&
  \mat{C}_{ii}\end{bmatrix}$, there is a non zero element located in the block
$\mat{C}_{ii}$. As any non-zero column of $\begin{bmatrix} \mat{C}_{i1}&\dots&\mat{C}_{i,i-1}\end{bmatrix}$
  has a leading coefficient in $\mat{A}$ at a row index stricly lower than $i$,
  there can not be more than $s-1$ of them.
  These $s-1$ columns of  $\begin{bmatrix} \mat{C}_{i1}&\dots&\mat{C}_{i,i-1}\end{bmatrix}$ can all be gathered in the block
$\mat{C}_{i,i-1}$ of column dimension~$s$.

There only remains to show that $\mat{S}\mat{T}$ is the matrix $\mat{C}$ of
step~\ref{step:CminusD}. For every pair of indices $(j,k)$ selected in
loop~\ref{step:loop}, right multiplication by
$(\mat{I}_n+\mat{\Delta}^{(k,j)}-\mat{\Delta}^{(k,k)})$ adds up column $k$ to column
$j$ and zeroes out column $k$. On matrix $\mat{S}$, this has the effect of
reverting each operation done at step~\ref{step:movecol} in the reverse order
of  the loop~\ref{step:loop}.
  \end{proof}

\begin{proposition}
  If an $s$-overlapping matrix $\mat{A}$ is in column echelon form, then, the
  structured representation $(\mat{D},\mat{S},\mat{T},\mat{P})$ is such that
  $\mat{P}=\mat{I}_r$ and $k_i\geq s \ \forall i<t$.
\end{proposition}
\begin{proof}
The leading elements of each column are already sorted in a column echelon
form, hence $\mat{P}=\mat{I}_r$.
Then, each  block $\mat{C}_{ii}$ contains $s$ pivots, hence $k_i\geq s$.
\end{proof}


We can now define the compact Bruhat representation.
\begin{definition}\label{def:CB}
  The compact Bruhat representation of an $n\times n$ $s$-quasiseparable left
  triangular matrix $\mat{A}$ is given by the tuples
  $(\mat{D}^{\mat{C}_\mat{A}},\mat{S}^{\mat{C}_\mat{A}},\mat{T}^{\mat{C}_\mat{A}})$,
  $(\mat{D}^{\mat{E}_\mat{A}},\mat{S}^{\mat{E}_\mat{A}},\mat{T}^{\mat{E}_\mat{A}})$
where $\mat{D}^{\mat{C}_\mat{A}}, \mat{S}^{\mat{C}_\mat{A}},(\mat{D}^{\mat{E}_\mat{A}})^T$ and $(\mat{S}^{\mat{E}_\mat{A}})^T$ are $n\times r$
block diagonal, with blocks of column dimension $s$, and 
$\mat{T}^{\mat{C}_\mat{A}}$ and $(\mat{T}^{\mat{E}_\mat{A}})^T$ are lower triangular
$\{0,1\}$-matrices with $r$ coefficients equals to $1$ placed on distinct rows,
and a permutation matrix $\mat{R}^\mat{A}$
such that
$$
\left\{
\begin{array}{lll}
  \mat{C}^\mat{A} &=&\mat{D}^{\mat{C}_\mat{A}} + \mat{S}^{\mat{C}_\mat{A}} \mat{T}^{\mat{C}_\mat{A}},\\
  \mat{E}^\mat{A} &=&\mat{D}^{\mat{E}_\mat{A}} +\mat{T}^{\mat{E}_\mat{A}}\mat{S}^{\mat{E}_\mat{A}} \\
\end{array}
\right.
$$
and $ \mat{A} = \LTP(\mat{C}^\mat{A} \mat{R}^\mat{A}\mat{E}^\mat{A})$ is a generalized Bruhat decomposition of $\mat{A}$.
\end{definition}
\begin{figure}[htbp]
  \centering
  \begin{minipage}{.4\textwidth}
  \includegraphics[height=\myfigsize]{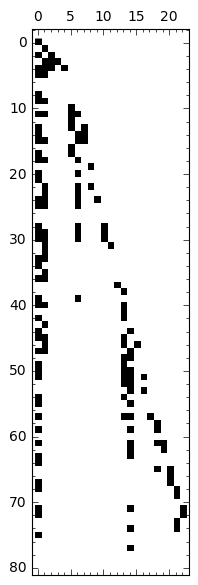}
   \includegraphics[height=\myfigsize]{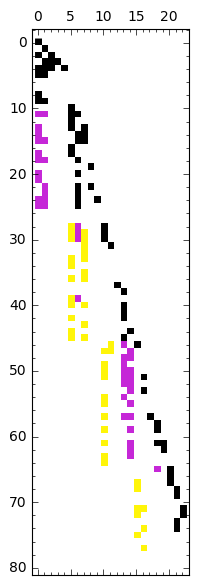}
  \end{minipage}
  \hfill
    \begin{minipage} {.55\textwidth}
     \includegraphics[width=\myfigsize]{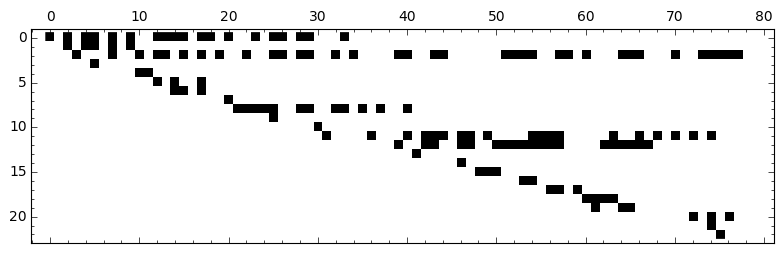}\\
     \includegraphics[width=\myfigsize]{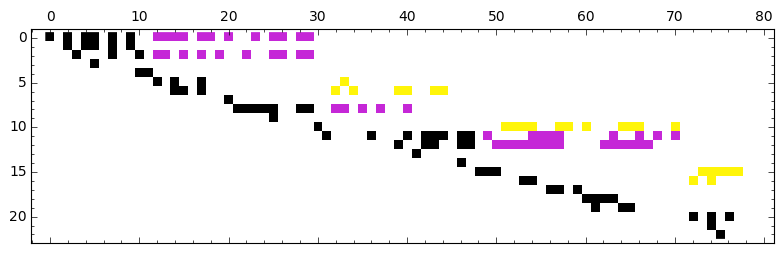}
    \end{minipage}
  \caption{Support of the matrices $\mat{C}=\mathcal{L}\mathcal{Q}$ (left),
    $\mat{E}=\mathcal{P}\mathcal{U}$ (top right) of the $s$-overlapping CRE
    decomposition of Theorem~\ref{th:genbruhatoverlap} applied to the matrix of
    Figure~\ref{fig:bruhatstorage}. The Compression to the block bi-diagonal
    structure of the corresponding compact Bruhat generator is shown in the
    central and bottom right matrices. There,
 $\mat{D}$ is in black and
  $\mat{S}$ in magenta and yellow; those rows and columns 
  moved at step~\ref{step:movecol} of Algorithm~\ref{alg:compress} are
   in yellow.}
  \label{fig:compactbruhat}
\end{figure}
%
\section{Computing with RRR representations}
\label{sec:RRR}

In this section, we will keep considering that the RRR representation
is based on any rank revealing factorization (RR), which could originate from
various matrix factorizations: PLUQ, CUP, PLE, QR, SVD, etc.
We will assume that there exists an algorithm \texttt{RRF} computing such a rank
revealing factorization. For instance, PLUQ, CUP, PLE decomposition algorithms can be used
to compute such a factorization in time $T_\texttt{RRF}(m,n,r)=\GO{mnr^{\omega-2}}$ on an $m\times n$ matrix of rank $r$~\citep{JPS13}.

\subsection{Construction of the generator}

The construction of the RRR representation simply consists in computing rank revealing
factorizations of all off-diagonal submatrices in a binary splitting of the main
diagonal.
Let $T_\texttt{RRR}(n,s)$ denote the cost of the computation of the  binary tree generator
for an $n\times n$ matrix of order of quasiseparability $s$.
It satisfies the recurrence relation
$T_\texttt{RRR}(n,s) = 2T_\texttt{RRF}(n,n,s) + 2T_\texttt{RRR}(n/2,s)$
which solves in 
$
T_\texttt{RRR}(n,s) = \GO{ s^{\omega-2}n^2 }
$.

\subsection{Matrix-vector product}
In the RRR representation, the application of a vector to the
quasiseparable matrix takes the same amount of field
operations as the number of coefficients used for its representation. This yields a
cost of $\GO{n(r_L \log\frac{n}{r_L} + r_U \log \frac{n}{r_U})} = \GO{sn\log\frac{n}{s}}$ field
operations.

\subsection{Auxiliary algorithms}
In the following, we present a set of routines that will be used to build
multiplication and inversion algorithms for RRR representations.
Algorithm~\ref{alg:RRRExpand} expands a matrix from an RRR
representation to a dense representation.
\begin{algorithm}[htbp]
{
  \caption{\texttt{RRRExpand}: expands an RRR representation into a dense representation} \label{alg:RRRExpand}
  \begin{algorithmic}[1]
    \Require{$\mat{A}$, an $n\times n$ $s$-quasiseparable matrix in an RRR representation,}
    \Ensure{$\mat{B} \leftarrow \mat{A}$ in a dense representation.}
    \If{$n\leq s$}
     \Return $\mat{B}\leftarrow \mat{A}$
    \EndIf
    \State $\mat{B}_{11}\leftarrow \texttt{RRRExpand}(\mat{A}_{11})$
    \State $\mat{B}_{22}\leftarrow \texttt{RRRExpand}(\mat{A}_{22})$
    \State $\mat{B}_{12}\leftarrow \mat{L}^\mat{A}_{12} \mat{R}^\mat{A}_{12}$
    \State $\mat{B}_{21}\leftarrow \mat{L}^\mat{A}_{21} \mat{R}^\mat{A}_{21}$
    \State \Return $\mat{B}\leftarrow 
    \begin{bmatrix}
      \mat{B}_{11}& \mat{B}_{12}\\
      \mat{B}_{21}& \mat{B}_{22}
    \end{bmatrix}
    $
\end{algorithmic}
}
\end{algorithm}
The recurring relation
$T_\texttt{RRRExpand}(n,s) =2T_\texttt{RRRExpand}(n/2,s)+O(n^2s^{\omega-2})$ for $ n>s$
yields
directly $$T_\texttt{RRRExpand}(n,s) = O(n^2s^{\omega-2}).$$

Algorithm~\ref{alg:RRxRR} multiplies two rank revealing factorizations and
outputs the result in a rank revealing factorization.
As $\mat{L}^{\mat{A}}$ and $\mat{L}^{\mat{X}}$ have full column rank, so is
their product. Hence the $\mat{L}^{\mat{C}}\mat{R}^{\mat{C}}$ is a rank
revealing factorization of the product.
\begin{algorithm}[htbp]
{
  \caption{\texttt{RRxRR}: multiplies two matrices stored as rank revealing factorization} \label{alg:RRxRR}
  \begin{algorithmic}[1]
    \Require{$\mat{A}$, an $m\times k$ matrix of rank $\leq s$ in an RR representation $\mat{L}^\mat{A}\times \mat{R}^\mat{A}$,}
    \Require{$\mat{B}$, an $k\times n$ matrix of rank $\leq t$ in an RR representation  $\mat{L}^\mat{B}\times \mat{R}^\mat{B}$,}
    \Ensure{$\mat{C} \leftarrow \mat{A}\times\mat{B}$ in an RR representation  $\mat{L}^\mat{C}\times \mat{R}^\mat{C}$.}
      \State $\mat{X}\leftarrow \mat{R}^\mat{A} \mat{L}^\mat{B}$
    \State $(\mat{L}^\mat{X},\mat{R}^\mat{X}) \leftarrow \texttt{RRF}(\mat{X})$\label{step:RRF1}
    \Comment {Computes the RR factorization $\mat{X}=\mat{L}^\mat{X}\times \mat{R}^\mat{X}$}
    \State $\mat{L}^\mat{C} \leftarrow \mat{L}^\mat{A}\mat{L}^\mat{X}$
    \State $\mat{R}^\mat{C} \leftarrow \mat{R}^\mat{X}\mat{R}^\mat{B}$
  \end{algorithmic}
}
\end{algorithm}
The resulting cost (assuming $s\leq t$ without loss of generality) is
$$ T_\texttt{RRxRR}(m,k,n,s,t) =\GO{s^{\omega-2}tk}+T_\texttt{RRF}(s,t)+
\GO{r_\mat{X}^{\omega-2}(ms+nt)}.$$
With $n=\Theta(m)=\Theta(k)$, this is $  T_\texttt{RRxRR}(n,s,t) =\GO{(s+t)^{\omega-1}n}.$


Algorithm~\ref{alg:RR+RR} adds two rank revealing factorizations.
It first stacks together the left sides and the right sides  of the rank
revealing factorizations of the two terms. The resulting factorization may not
reveal the rank as the inner dimension may be larger.
Therefore, a rank revealing factorization of each factor is first computed,
before invoquing \texttt{RRxRR} to obtain an RR representation of their product.
\begin{algorithm}[htbp]
{
  \caption{\texttt{RR+RR}: adds two matrices stored as rank revealing factorization} \label{alg:RR+RR}
  \begin{algorithmic}[1]
    \Require{$\mat{A}$, an $m\times n$ matrix of rank $\leq s$ in an RR representation $\mat{L}^\mat{A}\times \mat{R}^\mat{A}$}
    \Require{$\mat{B}$, an $m\times n$ matrix of rank $\leq t$ in an RR representation $\mat{L}^\mat{B}\times \mat{R}^\mat{B}$}
    \Ensure{$\mat{D} \leftarrow \mat{A}+\mat{B}$ in an RR representation
      $\mat{L}^\mat{D}\times \mat{R}^\mat{D}$.}
    \State $\mat{X} \leftarrow \begin{bmatrix} \mat{L}^\mat{A} & \mat{L}^\mat{B} \end{bmatrix}$; 
    $\mat{Y} \leftarrow \begin{bmatrix} \mat{R}^\mat{A} \\ \mat{R}^\mat{B} \end{bmatrix}$

    \State $(\mat{L}^\mat{X},\mat{R}^\mat{X}) \leftarrow \texttt{RRF}(\mat{X})$\label{step:RRF2}
    \Comment {$\mat{X}=\mat{L}^\mat{X}\times \mat{R}^\mat{X}$,
      $r_\mat{X}=\text{rank}(\mat{X})$; $\mat{L}^\mat{X}$ is $m\times r_\mat{X}$
      and $\mat{R}^\mat{X}$ is  $r_\mat{X}\times n$}
    \State $(\mat{L}^\mat{Y},\mat{R}^\mat{Y}) \leftarrow \texttt{RRF}(\mat{Y})$ \label{step:RRF3}
    \Comment {$\mat{Y}=\mat{L}^\mat{Y}\times \mat{R}^\mat{Y}$,
      $r_\mat{Y}=\text{rank}(\mat{Y})$; $\mat{L}^\mat{Y}$ is $m\times r_\mat{Y}$
      and $\mat{R}^\mat{Y}$ is $r_\mat{Y} \times n$}
    \State $\mat{D}\leftarrow \texttt{RRxRR}(\mat{X} ,\mat{Y})$
    \Comment{Computes an RR representation of the product $\mat{D}=\mat{X}\mat{Y}$}
   \end{algorithmic}
}
\end{algorithm}
Assuming $n=\Theta(m)$, the time complexity is
$$ T_\texttt{RR+RR}(n,s,t) = 2T_\texttt{RRF}(n,s+t)+\GO{(s+t)^{\omega-1}n} = \GO{(s+t)^{\omega-1}n}.$$

Algorithm~\ref{alg:RRR+RR} adds a quasiseparable matrix in RRR representation
with a matrix in RR representation.
\begin{algorithm}[htbp]
  \caption{\texttt{RRR+RR}: adds a quasiseparable matrix in RRR representation and a
    rank revealing factorization} \label{alg:RRR+RR}
  \begin{algorithmic}[1]
    \Require{$\mat{A}$, an $n\times n$ $s$-quasiseparable matrix in an RRR representation}
    \Require{$\mat{B}$, an $n\times n$ matrix of rank $\leq t$ in an RR representation $\mat{L}^\mat{B}\times \mat{R}^\mat{B}$}
    \Ensure{$\mat{C} \leftarrow \mat{A} + \mat{B}$  in an RRR representation.}
    \If{$n\leq s+t$}
      \State \Return $\mat{C}\leftarrow$ $\texttt{RRRExpand}(\mat{A}) +  \mat{L}^\mat{B}\times \mat{R}^\mat{B}$
    \EndIf
    \State Split the matrices as
    $ \begin{bmatrix}\mat{C}_{11}& \mat{C}_{12}\\ \mat{C}_{21}&
      \mat{C}_{22} \end{bmatrix} \leftarrow  \begin{bmatrix}\mat{A}_{11}&
      \mat{A}_{12}\\ \mat{A}_{21}& \mat{A}_{22} \end{bmatrix}
    +\begin{bmatrix}\mat{B}_{11}& \mat{B}_{12}\\ \mat{B}_{21}&
    \mat{B}_{22} \end{bmatrix}$.
    \State $\mat{C}_{11}\leftarrow \texttt{RRR+RR}(\mat{A}_{11}, \mat{B}_{11})$ \Comment{$\mat{C}_{11}\leftarrow \mat{A}_{11}+\mat{B}_{11}$}
    \State $\mat{C}_{22}\leftarrow \texttt{RRR+RR}(\mat{A}_{22}, \mat{B}_{22})$ \Comment{$\mat{C}_{11}\leftarrow \mat{A}_{22}+\mat{B}_{22}$}
    \State $\mat{C}_{12}\leftarrow \texttt{RR+RR}(\mat{A}_{12},\mat{B}_{12})$ \Comment{$\mat{C}_{12}\leftarrow \mat{A}_{12}+\mat{B}_{12}$}
    \State $\mat{C}_{21}\leftarrow \texttt{RR+RR}(\mat{A}_{21},\mat{B}_{21})$
    \Comment{$\mat{C}_{21}\leftarrow \mat{A}_{21}+\mat{B}_{21}$}
    \State \Return $\mat{C}\leftarrow \begin{bmatrix}\mat{C}_{11}& \mat{C}_{12}\\\mat{C}_{21}& \mat{C}_{22} \end{bmatrix}$
  \end{algorithmic}

\end{algorithm}
The time complexity satisfies the recurring relation 
$$
\left\{\begin{array}{llll}
  T_\texttt{RRR+RR}(n,s,t) &=& 2T_\texttt{RRR+RR}(n/2,s,t) +  2T_\texttt{RR+RR}(n/2,s,t) & \text{ for } n > s+t\\
                         &=& 2T_\texttt{RRR+RR}(n/2,s,t) +  O(n(s+t)^{\omega-1}) \\
  T_\texttt{RRR+RR}(n,s,t) &=& T_\texttt{RRRExpand}(s+t,s)+ O((s+t)^2t^{\omega-2}) & \text{ for } n \leq s+t
  \end{array}\right.
$$
which solves in
\begin{align*}
  T_\texttt{RRR+RR}(n,s,t) &=& O((s+t)^{\omega-1}n\log \frac{n}{s+t} +
  \frac{n}{s+t}(s+t)^2(s^{\omega-2}+t^{\omega-2}) \\ &=& O((s+t)^{\omega-1}n\log
  \frac{n}{s+t}).
\end{align*}

\subsection{Quasiseparable times tall and skinny}

\begin{algorithm}[htbp]
  \caption{\texttt{RRRxTS}: multiplies a quasiseparable matrix in RRR representation
    with a tall and skinny matrix} \label{alg:RRRxTS}
  \begin{algorithmic}[1]
    \Require{$\mat{A}$, an $n\times n$ $s$-quasiseparable matrix in RRR representation}
    \Require{$\mat{B}$, an $n\times t$ matrix}
    \Ensure{$\mat{C} \leftarrow \mat{A} \mat{B}$}
    \If{$n\leq s+t$}
      \State \Return $\mat{C} \leftarrow \texttt{RRRExpand}(\mat{A})\times  \mat{B}$
    \EndIf
    \State Split the matrices as
    $ \begin{bmatrix}\mat{C}_{1}\\ \mat{C}_{2}\end{bmatrix} \leftarrow  \begin{bmatrix}\mat{A}_{11}&\mat{A}_{12}\\ \mat{A}_{21}& \mat{A}_{22} \end{bmatrix}
    \begin{bmatrix}\mat{B}_{1}\\ \mat{B}_{2} \end{bmatrix} $.
    \State $\mat{C}_{1} \leftarrow \texttt{RRRxTS}(\mat{A}_{11},\mat{B}_1)$ \Comment{$\mat{C}_{1} \leftarrow \mat{A}_{11} \mat{B}_1$}
    \State $\mat{C}_{2} \leftarrow \texttt{RRRxTS}(\mat{A}_{22},\mat{B}_2)$ \Comment{$\mat{C}_{2} \leftarrow \mat{A}_{22} \mat{B}_2$}
    \State $\mat{X}\leftarrow \mat{R}^\mat{A}_{12} \mat{B_2}$
    \State $\mat{C}_1 \leftarrow \mat{C_1} + \mat{L}^\mat{A}_{12} \mat{X}$   \Comment{$\mat{C_1} \leftarrow \mat{C_1}+\mat{A}_{12}\mat{B}_2 $}
    \State $\mat{Y}\leftarrow \mat{R}^\mat{A}_{21} \mat{B_1}$
    \State $\mat{C}_2 \leftarrow \mat{C_2} + \mat{L}^\mat{A}_{21} \mat{Y}$   \Comment{$\mat{C_2} \leftarrow \mat{C_2}+\mat{A}_{21}\mat{B}_1 $}
    \State \Return $\mat{C}\leftarrow \begin{bmatrix} \mat{C}_1\\ \mat{C}_2 \end{bmatrix} $
  \end{algorithmic}
\end{algorithm}
Algorithm~\ref{alg:RRRxTS} multiplies an $s$-quasiseparable matrix of dimension
$n\times n$  in RRR representation by a tall and skinny matrix: an
$n\times t$ rectangular dense matrix with $t\leq n$.
%

Let $T_\texttt{RRRxTS}(n,s,t)$ denote its cost.
The recurring relation
$$
\left\{\begin{array}{llll}
T_\texttt{RRRxTS}(n,s,t) &=& 2T_\texttt{RRRxTS}(n/2,s,t) +
O(n\max(s,t)\min(s,t)^{\omega-2}) &\text{ for } n> s+t\\
T_\texttt{RRRxTS}(n,s,t) &=& T_\texttt{RRRExpand}(s+t,s)+ O((s+t)^2t^{\omega-2})& \text{ for } n \leq s+t
\end{array}
\right.
$$
yields
\begin{align*}
  T_\texttt{RRRxTS}(n,s,t) &=& O((s+t)^{\omega-1}n\log \frac{n}{s+t} + \frac{n}{s+t}(s+t)^2(s^{\omega-2}+t^{\omega-2})\\
  & =& O((s+t)^{\omega-1}n\log
  \frac{n}{s+t}).
\end{align*}

From this algorithm, follows Algorithm~\ref{alg:RRRxRR}, computing the product
of an $s$-quasi\-sepa\-rable matrix in RRR representation by a rank revealing
factorization.
Similarly as for Algorithm~\ref{alg:RRxRR}, $\mat{R}^{\mat{X}}$ and
$\mat{R}^\mat{B}$ have full row rank, so has their product, which ensures that
the factors $\mat{L}^{\mat{D}},\mat{R}^{\mat{D}}$ form a rank revealing
factorization of the result.
\begin{algorithm}[htbp]
  \caption{\texttt{RRRxRR}: multiplies a quasiseparable matrix in RRR representation
    with a rank revealing factorization} \label{alg:RRRxRR}
  \begin{algorithmic}[1]
    \Require{$\mat{A}$, an $n\times n$ $s$-quasiseparable matrix in RRR representation}
    \Require{$\mat{B}$, an $n\times m$ matrix of rank $\leq t$ in an RR representation $\mat{L}^\mat{B}\times \mat{R}^\mat{B}$}
    \Ensure{$\mat{D} \leftarrow \mat{A} \mat{B}$ in a rank revealing factorization $\mat{L}^\mat{D}\times \mat{R}^\mat{D}$}
    \State $\mat{X} \leftarrow \texttt{RRRxTS}(\mat{A},\mat{L}^\mat{B})$ \Comment{$\mat{X} \leftarrow \mat{A} \mat{L}^\mat{B}$}
    \State $ (\mat{L}^\mat{X},\mat{R}^\mat{X}) \leftarrow \texttt{RRF}(\mat{X})$ \label{step:RRRxRR:cup}
    \Comment{Computes the RR factorization $\mat{X} = \mat{L}^\mat{X}\times \mat{R}^\mat{X}$}
    \State $\mat{L}^\mat{D}\leftarrow \mat{L}^\mat{X}$
    \State $\mat{R}^\mat{D}\leftarrow \mat{R}^\mat{X}\mat{R}^\mat{B}$
  \end{algorithmic}
\end{algorithm}
Its time complexity is $$T_\texttt{RRRxRR}(n,s,t) =T_\texttt{RRRxTS}(n,s,t) +T_\texttt{RRF}(n,t)+ O(nt^{\omega-1}) =O((s+t)^{\omega-1}n\log\frac{n}{s+t}) $$

\subsection{Quasiseparable times Quasiseparable}
The product of an $s$-quasiseparable matrix by a $t$-quasiseparable matrix is an
$(s+t)$-quasiseparable matrix~\citep{EiGo99}.
Algorithm~\ref{alg:RRRxRRR}, calling
Algorithms~\ref{alg:RRxRR}, \ref{alg:RR+RR}, \ref{alg:RRR+RR}
and~\ref{alg:RRRxTS}, shows how to perform such a multiplication with the RRR representations.
\begin{algorithm}[htbp]
{
  \caption{\texttt{RRRxRRR}} \label{alg:RRRxRRR}
  \begin{algorithmic}[1]
    \Require{$\mat{A}$, an $n\times n$ $s$-quasiseparable matrix in an RRR representation,}
    \Require{$\mat{B}$, an $n\times n$ $t$-quasiseparable matrix in an RRR representation,}
    \Ensure{$\mat{C} \leftarrow \mat{A}\times\mat{B}$ in an RRR representation.}
    \If{$n\leq s+t$}
       \State \Return $\mat{C} \leftarrow \texttt{RRRExpand}(\mat{A}) \times \texttt{RRRExpand}(\mat{B})$
       
       \EndIf
    \State Split the matrices as
    $ \begin{bmatrix}\mat{C}_{11}& \mat{C}_{12}\\ \mat{C}_{21}&
      \mat{C}_{22} \end{bmatrix} \leftarrow  \begin{bmatrix}\mat{A}_{11}&
      \mat{A}_{12}\\ \mat{A}_{21}& \mat{A}_{22} \end{bmatrix}
    \begin{bmatrix}\mat{B}_{11}& \mat{B}_{12}\\ \mat{B}_{21}&
    \mat{B}_{22} \end{bmatrix}$.
    \State $\mat{C}_{11}\leftarrow$ \texttt{RRRxRRR} $(\mat{A}_{11},\mat{B}_{11})$ \Comment{$\mat{C}_{11} \leftarrow \mat{A}_{11} \mat{B}_{11}$}
    \State $\mat{C}_{22}\leftarrow$ \texttt{RRRxRRR} $(\mat{A}_{22},\mat{B}_{22})$ \Comment{$\mat{C}_{22} \leftarrow \mat{A}_{22} \mat{B}_{22}$}
    \State $\mat{X} \leftarrow$ \texttt{RRxRR} $(\mat{A}_{12},\mat{B}_{21})$ \Comment{$\mat{X} \leftarrow \mat{A}_{12} \mat{B}_{21}$}
    \State $\mat{Y} \leftarrow$ \texttt{RRxRR} $(\mat{A}_{21},\mat{B}_{12})$ \Comment{$\mat{Y} \leftarrow \mat{A}_{21} \mat{B}_{12}$}
    \State  $\mat{C}_{11} \leftarrow$ \texttt{RRR+RR} $(\mat{C}_{11},\mat{X})$ \Comment{$\mat{C}_{11} \leftarrow \mat{C}_{11} +\mat{X}$}\label{step:RRRxRRR:C11}
    \State $\mat{C}_{22} \leftarrow$ \texttt{RRR+RR} $(\mat{C}_{22},\mat{Y})$ \Comment{$\mat{C}_{22} \leftarrow \mat{C}_{22} +\mat{Y}$}\label{step:RRRxRRR:C22}
    \State $\mat{L}^\mat{X} \leftarrow \texttt{RRRxTS}(\mat{A}_{11},\mat{L}^\mat{B}_{12})$; $\mat{R}^\mat{X} \leftarrow \mat{R}^\mat{B}_{12}$
    \Comment{$\mat{X} \leftarrow \mat{A}_{11} \mat{B}_{12}$  in RR representation}
    \State $\mat{L}^\mat{Y} \leftarrow \mat{L}^\mat{A}_{12}$; $\mat{R}^\mat{Y} \leftarrow \texttt{TSxRRR}(\mat{R}^\mat{A}_{12},\mat{B}_{22})$
    \Comment{$\mat{Y} \leftarrow \mat{A}_{12} \mat{B}_{22}$  in RR representation
      revealing factorization}
    \State $\mat{C}_{12} \leftarrow \texttt{RR+RR}(\mat{X},\mat{Y})$ \Comment{$\mat{C}_{12} \leftarrow \mat{X}+\mat{Y}$}
    \State $\mat{L}^\mat{X} \leftarrow \texttt{RRRxTS}(\mat{A}_{11},\mat{L}^\mat{B}_{21})$; $\mat{R}^\mat{X} \leftarrow \mat{R}^\mat{B}_{21}$
    \Comment{$\mat{X} \leftarrow \mat{A}_{11} \mat{B}_{21}$ in RR representation}
    \State $\mat{L}^\mat{Y} \leftarrow \mat{L}^\mat{A}_{21}$; $\mat{R}^\mat{Y} \leftarrow \texttt{TSxRRR}(\mat{R}^\mat{A}_{21},\mat{B}_{22})$
    \Comment{$\mat{Y} \leftarrow \mat{A}_{21} \mat{B}_{22}$ in RR representation}
    \State $\mat{C}_{21} \leftarrow \texttt{RR+RR}(\mat{X},\mat{Y})$ \Comment{$\mat{C}_{21} \leftarrow \mat{X}+\mat{Y}$}
    \State \Return $\mat{C} \leftarrow 
    \begin{bmatrix}
      \mat{C}_{11}& \mat{C}_{12}\\
      \mat{C}_{21}& \mat{C}_{22}\\
    \end{bmatrix}$
\end{algorithmic}
}
\end{algorithm}
In steps~\ref{step:RRRxRRR:C11} and~\ref{step:RRRxRRR:C22}, a
$(s+t)$-quasiseparable matrix is added to a rank revealing factorization of rank
$(s+t)$. It should in general result in an RRR representation of an
$2(s+t)$-quasiseparable matrix. However, the matrix $\mat{C}$ is no more than
$(s+t)$-quasiseparable, hence the rank revealing factorization of the result, will
have rank only $s+t$. The reductions to RR representation, performed in
step~\ref{step:RRF1} of Algorithm~\ref{alg:RRxRR} and steps~\ref{step:RRF2}
and~\ref{step:RRF3} of Algorithm~\ref{alg:RR+RR}, ensure
that this factorization will be reduced to this size.

%
 Let  $T_\texttt{RRR$\times$RRR}(n,s,t)$ denote the time complexity of this
 algorithm.
%
 If $n\leq s+t$, then 
 $
    T_\texttt{RRRxRRR}(n,s,t) =
    T_\texttt{RRRExpand}(n,s)+T_\texttt{RRRExpand}(n,t) + O(n^\omega) 
    =O((s+t)^{\omega}).
    $
%
    Now consider the case $n>s+t$.
  \begin{align*}
    T_\texttt{RRRxRRR}(n,s,t) &=& 2 T_\texttt{RRRxRRR}(n/2,s,t) + 2T_\texttt{RRxRR}(n/2,s,t) + 2T_\texttt{RRRxTS}(n/2,s,t)   \\
    && + 2T_\texttt{RRR+RR}(n/2,s+t,s+t)+   2 T_\texttt{RR+RR}(n/2,s+t,s+t)\\
    &=& 2 T_\texttt{RRRxRRR}(n/2,s,t)  +O((s+t)^{\omega-1}n\log \frac{n}{s+t})
  \end{align*}

Consequently, $T_\texttt{RRRxRRR}(n,s,t)=O((s+t)^{\omega-1}n\log^2\frac{n}{s+t})$.

\subsection{Computing the inverse in RRR representation}

We consider the case, as in~\citep[\S~6]{EiGo99}, where the matrix to be inverted
has generic rank profile, i.e. all of its leading principal minors are
non-vanishing.
Under this assumption, Strassen's divide and conquer algorithm~\citep{Str69}
reduces the computation of the inverse to matrix multiplication.
More precisely, the inverse is recursively computed using the following block
$2\times 2$ formula:
$$
\begin{bmatrix}
  \mat{A}_{11} & \mat{A}_{12} \\
  \mat{A}_{21} & \mat{A}_{22} \\
\end{bmatrix}^{-1} 
= 
\begin{bmatrix}
  \mat{A}_{11}^{-1}+\mat{A}_{11}^{-1}\mat{A}_{12}\mat{D}^{-1}\mat{A}_{21}\mat{A}_{11}^{-1} & -\mat{A}_{11}^{-1}\mat{A}_{12}\mat{D}^{-1} \\
- \mat{D}^{-1}\mat{A}_{21}\mat{A}_{11}^{-1} & \mat{D}^{-1} \\
\end{bmatrix},
$$
where $\mat{D}=\mat{A}_{22}-\mat{A}_{21}\mat{A}_{11}^{-1}\mat{A}_{12}$.

This formula leads to a recursive algorithm that we adapt to the case of
quasiseparable matrices in RRR representation in algorithm~\ref{alg:RRR:inverse}.
\begin{algorithm}[htbp]
  \caption{\texttt{RRRinvert}: compute the inverse in RRR representation} \label{alg:RRR:inverse}
  \begin{algorithmic}[1]
    \Require{$\mat{A}$, an $n\times n$ $s$-quasiseparable strongly regular matrix in RRR representation,}
    \Ensure{$\mat{X} = \mat{A}^{-1}$, $s$-quasiseparable in RRR representation.}
    \If{$n\leq s$}
      \State $\mat{Y} \leftarrow \texttt{RRRExpand}(\mat{A})$
      \State \Return $\mat{X} \leftarrow \texttt{Invert}(\mat{Y})$
    \EndIf
    \State Split the matrix as
    $\mat{A}= \begin{bmatrix}\mat{A}_{11}&\mat{A}_{12}\\ \mat{A}_{21}&\mat{A}_{22} \end{bmatrix}$
    and $\mat{X}= \begin{bmatrix}\mat{X}_{11}&\mat{X}_{12}\\ \mat{X}_{21}&\mat{X}_{22} \end{bmatrix}$
    \State $\mat{Y}_{11} \leftarrow \texttt{RRRinvert}(\mat{A}_{11})$  \Comment{$\mat{Y}_{11}=\mat{A}_{11}^{-1}$}
    \State $\mat{Y}_{12} \leftarrow \texttt{RRRxRR}(\mat{Y}_{11},\mat{A}_{12})$ \Comment{$\mat{Y}_{12} \leftarrow \mat{A}_{11}^{-1}\mat{A}_{12}$}
    \State $\mat{Y}_{21} \leftarrow \texttt{RRRxRR}(\mat{A}_{21},\mat{Y}_{11})$ \Comment{$\mat{Y}_{21} \leftarrow \mat{A}_{21}\mat{A}_{11}^{-1}$}
    \State $\mat{Z}\leftarrow -\texttt{RRxRR}(\mat{A}_{21}\mat{Y}_{12})$ \Comment{$\mat{Z}\leftarrow -\mat{A}_{21}\mat{A}_{11}^{-1}\mat{A}_{12}$}
    \State $\mat{D}\leftarrow \texttt{RR+RR}(\mat{A}_{22},\mat{Z})$
    \Comment{$\mat{D}\leftarrow
      \mat{A}_{22}-\mat{A}_{21}\mat{A}_{11}^{-1}\mat{A}_{12}$} \label{step:invert:RR+RR}
    \State $\mat{X}_{22} \leftarrow \texttt{RRRinvert}(\mat{D})$ \Comment{$\mat{X}_{22} = \mat{D}^{-1}$}
    \State $\mat{X}_{21} \leftarrow -\texttt{RRRxRR}(\mat{X}_{22},\mat{Y}_{21})$ \Comment{$\mat{X}_{21} \leftarrow -\mat{D}^{-1}\mat{A}_{21}\mat{A}_{11}^{-1}$}
    \State $\mat{W} \leftarrow -\texttt{RRxRR}(\mat{Y}_{12},\mat{X}_{21})$ \Comment{$\mat{W} \leftarrow \mat{A}_{11}^{-1}\mat{A}_{12}\mat{D}^{-1}\mat{A}_{21}\mat{A}_{11}^{-1}$}
    \State $\mat{X}_{12} \leftarrow -\texttt{RRRxRR}(\mat{Y}_{12},\mat{X}_{22})$ \Comment{$\mat{X}_{12} \leftarrow -\mat{A}_{11}^{-1}\mat{A}_{12}\mat{D}^{-1}$}
    \State $\mat{X}_{11}\leftarrow \texttt{RRR+RR}(\mat{Y}_{11},\mat{W})$
    \Comment{$\mat{X}_{11}\leftarrow
      \mat{A}_{11}^{-1}+\mat{A}_{11}^{-1}\mat{A}_{12}\mat{D}^{-1}\mat{A}_{21}\mat{A}_{11}^{-1}
      $}
\State \Return $\mat{X} = 
\begin{bmatrix}
  \mat{X}_{11}&\mat{X}_{12}\\
  \mat{X}_{21}&\mat{X}_{22}\\
\end{bmatrix}$.
  \end{algorithmic}
\end{algorithm}
The fact that the inverse matrix $\mat{X}$ is itself $s$-quasiseparable, implies
that the matrix $\mat{D}$ is also $s$-quasiseparable and not
$2s$-quasiseparable, as the generic upper bound would say. The compression
happens in the \texttt{RR+RR} routine, at step~\ref{step:invert:RR+RR}. Hence
all operations except the recursive calls take $O(s^{\omega-1}n\log\frac{n}{s})$.
The overall complexity of Algorithm~\ref{alg:RRR:inverse} is therefore $T_\texttt{RRRInverse}(n,s)=O(s^{\omega-1}n\log^2\frac{n}{s})$.

\section{Computing with a Compact Bruhat representation}
\label{sec:CB}

\subsection{Construction of the generator}

We first propose in Algorithm~\ref{alg:bruhatgenerator} an evolution of Algorithm~\ref{alg:LTElim}
to compute the factors of the Bruhat generator  (without compression) for a left triangular matrix.
\begin{algorithm}[htbp]
{
  \caption{LT-Bruhat}
\label{alg:bruhatgenerator}
  \begin{algorithmic}[1]
    \Require{$\mat{A}$: an $n\times n$ matrix}
    \Ensure{$(\mathcal{L},\mathcal{R},\mathcal{U})$: a Bruhat generator for the left
      triangular part of $\mat{A}$}
    \State \textbf{if} $n=1$ \textbf{then} \textbf{return} $([0],[0],[0])$
    \State {Split $\mat{A} = \begin{bmatrix} \mat{A_{1}} & \mat{A_{2}}\\\mat{A_{3}} \end{bmatrix}$ 
      where $\mat{A_{3}}$ is $\lfloor \frac{n}{2} \rfloor \times \lfloor
      \frac{n}{2}\rfloor$}
    \State Decompose $\mat{A_{1}} =
    \mat{P_1} \begin{bmatrix}\mat{L_1}\\\mat{M_1}\end{bmatrix}\begin{bmatrix}\mat{U_1}&\mat{V_1}\end{bmatrix} \mat{Q_1}$ 
    \State $\mat{R_1}  \leftarrow \mat{P_1} \begin{bmatrix}
      \mat{I_{r_1}}\\&\mat{0}\end{bmatrix}\mat{Q_1}$ where $r_1=\text{rank}(\mat{A_1})$.
    \State $\begin{bmatrix} \mat{B_1}\\ \mat{B_2}\end{bmatrix}
    \leftarrow \mat{P_1}^T\mat{A_{2}}$ 
    \State $
    \begin{bmatrix}
      \mat{C_1}&\mat{C_2}
    \end{bmatrix}
    \leftarrow \mat{A_{3}}\mat{Q_1}^T$ 
    \State Here $A =
    \left[\begin{array}{cc|c}
        \mat{L_1} \backslash \mat{U_1}& \mat{V_1}& \mat{B_1}\\
        \mat{M_1}               & \mat{0}  & \mat{B_2}\\
        \hline
        \mat{C_1}               & \mat{C_2}& \\
      \end{array}\right]$. \label{step:partialPLUQ}
    \State $\mat{D}\leftarrow \mat{L_1}^{-1}\mat{B_1}$ 
    \State $\mat{E}\leftarrow \mat{C_1}\mat{U_1}^{-1}$ 
    \State $\mat{F}\leftarrow \mat{B_2}-\mat{M_1}\mat{D}$ 
    \State $\mat{G}\leftarrow \mat{C_2}-\mat{E}\mat{V_1}$ 
    \State Here $\mat{A}=
    \left[\begin{array}{cc|c}
        \mat{L_1} \backslash \mat{U_1}& \mat{V_1}& \mat{D}\\
        \mat{M_1}               & \mat{0}  & \mat{F}\\
        \hline
        \mat{E}               & \mat{G}& \\
      \end{array}\right]$.
    
    \State $\mat{H} \leftarrow \mat{P_1}  \begin{bmatrix} \mat{0}_{r_1 \times \frac{n}{2}} \\ \mat{F} \end{bmatrix}$
    \State $\mat{I} \leftarrow  \begin{bmatrix} \mat{0}_{r_1 \times \frac{n}{2}} & \mat{G} \end{bmatrix} \mat{Q_1}$
    \State\label{step:firstLT} $(\mathcal{L}_2,\mathcal{R}_2,\mathcal{U}_2) \leftarrow
    \texttt{LT-Bruhat}(\mat{H})$ 
    \State\label{step:secondLT} $(\mathcal{L}_3,\mathcal{R}_3,\mathcal{U}_3) \leftarrow
    \texttt{LT-Bruhat}(\mat{I})$ 
    \State $\mathcal{L} \leftarrow
    \begin{bmatrix}
      \mat{P}_1 
      \begin{bmatrix}
        \mat{L}_1&\mat{0} \\
        \mat{M}_1 &\mat{0}\\
      \end{bmatrix}
      \mat{Q}_1 & \mat{0}
      \\
      \LTP (
      \begin{bmatrix}
        \mat{E} & \mat{0}      
      \end{bmatrix}
      \mat{Q}_1 ) & \mat{0}
      \end{bmatrix}
+
\begin{bmatrix}  \mat{0}&\mathcal{L}_2\\\mathcal{L}_3\end{bmatrix}
$
    \State $\mathcal{U} \leftarrow
    \begin{bmatrix}
      \mat{P}_1
      \begin{bmatrix}\mat{U}_1&V_1\\\mat{0}&\mat{0} \end{bmatrix} \mat{Q}_1 & \LTP(\mat{P}_1
      \begin{bmatrix} \mat{D}\\\mat{0}  \end{bmatrix}) \\
      \mat{0}&\mat{0}
    \end{bmatrix}
+
\begin{bmatrix}  \mat{0}&\mathcal{U}_2\\\mathcal{U}_3\end{bmatrix}
$\label{step:calU}
    \State $\mathcal{R} \leftarrow  \begin{bmatrix}     \mathcal{R}_1 & \mathcal{R}_2\\ \mathcal{R}_3   \end{bmatrix}$
   \State \textbf{return} $(\mathcal{L},\mathcal{R},\mathcal{U})$
  \end{algorithmic}
}
\end{algorithm}

\begin{theorem}
  For any $n\times n$ matrix $\mat{A}$ with a left triangular part of
  quasiseparable order $s$, Algorithm~\ref{alg:bruhatgenerator} computes the
  Bruhat generator of the
  left triangular part of $\mat{A}$ in $\GO{s^{\omega-2}n^2}$ field operations.
\end{theorem}
\begin{proof}
The correctness of $\mathcal{R}$ is proven in Theorem~\ref{th:LTRPM}.
We will prove by induction the correctness of $\mathcal{U}$, noting that the correctness of
$\mathcal{L}$ works similarly.

Let $\mat{H}=\mat{P}_2\mat{L}_2\mat{U}_2\mat{Q}_2$ and
$\mat{I}=\mat{P}_3\mat{L}_3\mat{U}_3\mat{Q}_3$ be PLUQ decompositions of
$\mat{H}$ and $\mat{I}$ revealing their rank profile matrices. 
Assume that Algorithm LT-Bruhat is correct in the two recursive calls~\ref{step:firstLT}
and~\ref{step:secondLT}, that is 
$$
\begin{array}{ll}
\mathcal{U}_2 = \LTP(\mat{P}_2\begin{bmatrix} \mat{U}_2\\\mat{0}\end{bmatrix} \mat{Q}_2), & 
\mathcal{U}_3 = \LTP(\mat{P}_3\begin{bmatrix} \mat{U}_3\\\mat{0}\end{bmatrix} \mat{Q}_3),\\
\mathcal{L}_2 = \LTP(\mat{P}_2\begin{bmatrix} \mat{L}_2&\mat{0}\end{bmatrix} \mat{Q}_2), &
\mathcal{L}_3 = \LTP(\mat{P}_3\begin{bmatrix} \mat{L}_3&\mat{0}\end{bmatrix} \mat{Q}_3).\\
\end{array}
$$

At step~\ref{step:partialPLUQ}, we have
{
$$
\begin{bmatrix}
  \mat{A}_1 & \mat{A}_2\\
  \mat{A}_3 &*
\end{bmatrix}
 =
\begin{bmatrix}\mat{P}_1\\&\mat{I}_{\frac{n}{2}}\end{bmatrix}
\left[\begin{array}{cc|c}
  \mat{L}_1&&\\ 
  \mat{M}_1 &\mat{I}_{\frac{n}{2}-r_1}\\
\hline
  \mat{E} & \mat{0}&\mat{I}_{\frac{n}{2}}
\end{array}\right]
\left[\begin{array}{cc|c}
  \mat{U}_1 & \mat{V}_1 &  \mat{D}\\
&\mat{0} &  \mat{F}\\
\hline
&\mat{G}&
\end{array}
\right]
\begin{bmatrix} \mat{Q}_1\\ & \mat{I}_{\frac{n}{2}} \end{bmatrix}
$$
}
As the first $r_1$ rows of $\mat{P}_1^T\mat{H}$ are zeros, there exists $\mat{\bar P}_2$ a
permutation matrix and $\mat{\bar L}_2$, a lower triangular matrix, such that
$\mat{P}_1^T\mat{P}_2\mat{L}_2 = 
\begin{bmatrix} \mat{0}_{r_1\times \frac{n}{2}} \\ \mat{\bar P}_2 \mat{\bar L}_2 \end{bmatrix}$. 
Similarly, there exsist $\mat{\bar Q}_3$, a
permutation matrix and $\mat{\bar U}_3$, an upper triangular matrix, such that $\mat{U}_3\mat{Q}_3\mat{Q}_1^T =
\begin{bmatrix}
  \mat{0}_{\frac{n}{2}\times r_1} &
  \mat{\bar U}_3 \mat{\bar Q}_3
\end{bmatrix}$.
Hence
{
  $$
  \begin{bmatrix}
  \mat{A}_1 & \mat{A}_2\\
  \mat{A}_3 &*
\end{bmatrix}
 =
\begin{bmatrix}\mat{P}_1
\\&\mat{P}_3\end{bmatrix}
\left[\begin{array}{cc|c}
  \mat{L}_1&&\\ 
  \mat{M}_{1} & \mat{\bar P}_2\mat{\bar L}_2\\
\hline
  \mat{P}_3^T\mat{E} & \mat{0}&\mat{L}_3
\end{array}\right] 
\left[\begin{array}{cc|c}
  \mat{U}_1 & \mat{V}_1 &  \mat{D}\mat{Q}_2^T\\
&\mat{0} &  \mat{U_2}\\
&\mat{\bar U}_3\mat{\bar Q}_3
\end{array}
\right]
\begin{bmatrix} \mat{Q}_1\\ & \mat{Q}_2 \end{bmatrix}\\
$$
}
Setting $\mat{N}_1 = \mat{\bar P}_2^T\mat{M}_1$ and $\mat{W}_1 = \mat{V}_1
\mat{\bar Q}_3^T$, we have
{
$$
\begin{bmatrix}
  \mat{A}_1 & \mat{A}_2\\
  \mat{A}_3 &*
\end{bmatrix}
=
\begin{bmatrix}\mat{P}_1
  \begin{bmatrix}
    \mat{I}_{r_1}\\&\mat{\bar P}_2
  \end{bmatrix}
\\&\mat{P}_3\end{bmatrix}
\left[\begin{array}{ccc}
  \mat{L}_1\\ 
  \mat{N}_{1} & \mat{\bar L}_2\\
\hline
  \mat{E} &\mat{0}& \mat{\mat{L_3}}
\end{array}\right] 
\left[\begin{array}{cc|c}
  \mat{U}_1 & \mat{W}_1 &  \mat{D}\mat{Q}_2^T\\
&\mat{0} &  \mat{U_2}\\
&\mat{\bar U}_3
\end{array}
\right]
\begin{bmatrix}
  \begin{bmatrix}
    \mat{I}_{r_1}\\&\mat{\bar Q}_3
  \end{bmatrix}
\mat{Q}_1\\ & \mat{Q}_2 \end{bmatrix}.
$$
}
A PLUQ of $
\begin{smatrix}
  \mat{A}_1 & \mat{A}_2\\ \mat{A}_3
\end{smatrix}
$ revealing its rank profile matrix is then obtained from this decomposition by
a row block cylic-shift on the second factor and a column block cyclic shift on
the third factor as in \citep[Algorithm~1]{DPS13}.

Finally,
{
\begin{align*}
\mat{P} \begin{bmatrix}  \mat{U}\\0\end{bmatrix}\mat{Q}
&=&
\begin{bmatrix}  \mat{P}_1\\&\mat{I}_{\frac{n}{2}}\end{bmatrix}
\begin{bmatrix}
  \mat{U}_1 & \mat{V}_1 & \mat{D} \\
  & \mat{0} & \mat{\bar P}_2 \mat{U}_2\mat{Q}_2 \\
  &\mat{P}_3 \mat{\bar U}_3  \mat{\bar Q}_3\\
\mat{0}& \mat{0}&\mat{0}
\end{bmatrix}
\begin{bmatrix}
  \mat{Q}_1\\&\mat{I}_{\frac{n}{2}}
\end{bmatrix}\\
&=&
\begin{bmatrix}
\mat{P}_1  \begin{bmatrix}
      \mat{U}_1 & \mat{V}_1\\
      \mat{0}&\mat{0}
  \end{bmatrix}
\mat{Q}_1 & \mat{P}_1
\begin{bmatrix}
  \mat{D}\\\mat{0}
\end{bmatrix}
 \\
  \mat{0}& \mat{0} \
\end{bmatrix}
+\begin{bmatrix}
& \mat{P}_2
\begin{bmatrix}
  \mat{U}_2\\ \mat{0}
\end{bmatrix}
 \mat{Q}_2\\
\mat{P}_3
\begin{bmatrix}
  \mat{U}_3\\\mat{0}
\end{bmatrix}
\mat{Q}_3
\end{bmatrix}.
\end{align*}
}
Hence
{
$\LTP(\mat{P}\mat{U}\mat{Q}) =
\begin{bmatrix}
 \mat{P}_1
 \begin{bmatrix}
   \mat{U}_1 & \mat{V}_1 \\\mat{0}&\mat{0}
 \end{bmatrix}\mat{Q}_1
 & \LTP(\mat{P}_1
 \begin{bmatrix} \mat{D}\\\mat{0} \end{bmatrix}) \\
   \mat{0} & \mat{0} \\
\end{bmatrix} 
+ \begin{bmatrix}
&  \mathcal{U}_2\\
\mathcal{U}_3
\end{bmatrix}.
$
}


  
The complexity analysis is exactly that of Theorem~\ref{th:LTRPM}.
\end{proof}

The computation of a compact Bruhat generator, as shown in
Algorithm~\ref{alg:compactbruhatgenerator}, is then directly obtained by combining
Algorithm~\ref{alg:bruhatgenerator} with Algorithm~\ref{alg:compress}.

\begin{algorithm}[htbp]
  \caption{Compact Bruhat generator}
\label{alg:compactbruhatgenerator}
  \begin{algorithmic}[1]
    \Require{$\mat{A}$: an $n\times n$ left triangular matrix of quasiseparable order $s$}
    \Ensure{ $(\mat{D}^{{\mat{C}_\mat{L}}},\mat{S}^{{\mat{C}_\mat{L}}},\mat{T}^{{\mat{C}_\mat{L}}}), \mat{R}_\mat{L},(\mat{D}^{{\mat{E}_\mat{L}}},\mat{S}^{{\mat{E}_\mat{L}}},\mat{T}^{{\mat{E}_\mat{L}}})$ : a Compact Bruhat generator for  $\mat{L}=\LTP(\mat{J}_n\mat{A})$}
    \Ensure{ $(\mat{D}^{{\mat{C}_\mat{U}}},\mat{S}^{{\mat{C}_\mat{U}}},\mat{T}^{{\mat{C}_\mat{U}}}), \mat{R}_\mat{U},(\mat{D}^{{\mat{E}_\mat{U}}},\mat{S}^{{\mat{E}_\mat{U}}},\mat{T}^{{\mat{E}_\mat{U}}})$ : a Compact Bruhat generator for  $\mat{U}=\LTP(\mat{A}\mat{J}_n)$}
\State $\mat{L} \leftarrow \LTP(\mat{J}_n\mat{A})$
\State $\mat{U} \leftarrow \LTP(\mat{A}\mat{J}_n)$
\State $(\mathcal{L}_\mat{L},\mathcal{R}_\mat{L},\mathcal{U}_\mat{L})\leftarrow\texttt{LT-Bruhat}(\mat{L})$
\State $(\mathcal{L}_\mat{U},\mathcal{R}_\mat{U},\mathcal{U}_\mat{U})\leftarrow\texttt{LT-Bruhat}(\mat{U})$
\State $(\mat{D}^{{\mat{C}_\mat{L}}},\mat{S}^{{\mat{C}_\mat{L}}},\mat{T}^{{\mat{C}_\mat{L}}},\mat{P}^{\mat{C}_\mat{L}})\leftarrow \texttt{Compress-to-Block-Bidiagonal}(\mathcal{L}_\mat{L})$
\State $(\mat{D}^{{\mat{E}_\mat{L}}},\mat{S}^{{\mat{E}_\mat{L}}},\mat{T}^{{\mat{E}_\mat{L}}},\mat{P}^{\mat{E}_\mat{L}})\leftarrow (\texttt{Compress-to-Block-Bidiagonal}(\mathcal{U}_\mat{L}^T))^T$
\State $\mat{R}_\mat{L}\leftarrow \begin{bmatrix}  \mat{I}_r &
  \mat{0}\end{bmatrix} \mat{P}^{\mat{C}_\mat{L}}
\mathcal{R}_\mat{L}^T\mat{P}^{\mat{E}_\mat{L}} \begin{bmatrix}  \mat{I}_r \\ \mat{0}\end{bmatrix}$
\State $\mat{R}_\mat{U}\leftarrow \begin{bmatrix}  \mat{I}_r &
  \mat{0}\end{bmatrix} \mat{P}^{\mat{C}_\mat{U}}
\mathcal{R}_\mat{U}^T\mat{P}^{\mat{E}_\mat{U}} \begin{bmatrix}  \mat{I}_r \\ \mat{0}\end{bmatrix}$

  \end{algorithmic}
\end{algorithm}

\subsection{Multiplication by a tall and skinny matrix}
\label{sec:CBxTS}
We consider the multiplication of an $s$-quasiseparable matrix in Compact
Bruhat representation by an $n\times t$ dense rectangular matrix ($t\leq s$),
and show that is can be performed in $\GO{st^{\omega-2}n}=\GO{s^{\omega-1}n}$ field operations.

The Compact Bruhat representation stores a representation of two left triangular
matrices, corresponding to the upper and lower triangular parts of the
matrix. Hence it suffices to show how to multiply an $s$-quasiseparable left triangular matrix
in Compact Bruhat representation with a tall and skinny matrix.

Using the Definition~\ref{def:CB}, this means computing
$$
\mat{C}=\text{Left}(\mat{C}^\mat{A}\mat{R}^\mat{A}\mat{E}^\mat{A})\mat{B}
$$
where $B$ is dense $n\times t$.
Without the \texttt{Left} operator, the target complexity $O(s^{\omega-1}n)$ would
be reached by first computing the product $\mat{E}^\mat{A}\mat{B}$ and then
applying $\mat{R}^\mat{A}$ and $\mat{C}^\mat{A}$ on the left.
However because of the \texttt{Left} operator, each row of the result matrix
$\mat{C}$ involves a distinct partial sum of the product
$\mat{E}^\mat{A}\mat{B}$: $$\mat{C}_{i,*} = \mat{C}^\mat{A}_{i,*}\mat{R}^\mat{A} \left(\sum_{j=1}^{n-i}\mat{E}^\mat{A}_{*,j}\mat{B}_{j,*}\right).$$
We will therefore avoid computing the accumulation in this product, keeping
point-wise products available in memory. In order to reach the target
complexity, the products of dimension $s$ will be computed with accumulation,
keeping the terms of the unevaluated sum available at the level of  size $s$ blocks.

Cutting these matrices on a grid of  size $s$, let $N=\lceil n/s\rceil$ and
$\mat{C}=\begin{bmatrix}  \mat{C}_1&\dots&\mat{C}_{N}\end{bmatrix}^T$,
$\mat{E}^\mat{A}=\begin{bmatrix}\mat{E}^\mat{A}_1&\dots&\mat{E}^\mat{A}_{N}\end{bmatrix}$,
$\mat{C}^\mat{A}=\begin{bmatrix}\mat{C}^\mat{A}_1&\dots&\mat{C}^\mat{A}_{N}\end{bmatrix}^T$ 
and
$\mat{B}=\begin{bmatrix}  \mat{B}_1&\dots&\mat{B}_{N}\end{bmatrix}^T$.
We have
$$
\mat{C}_i = \mat{C}^\mat{A}_i\mat{R}^\mat{A}\sum_{j=1}^{N -i}\mat{E}^\mat{A}_j
\mat{B}_j + \text{Left}\left(\mat{C}^\mat{A}_i\mat{R}^\mat{A}\mat{E}^\mat{A}_{N - i + 1}\right)\mat{B}_{N - i + 1}.
$$
Each of these blocks $\mat{C}_i$ are then computed as shown in Algorithm~\ref{alg:LeftCBxTS}.
\begin{algorithm}[htbp]
   \caption{\texttt{LeftCBxTS} }
   \label{alg:LeftCBxTS}
\begin{algorithmic}[1]
     \Require{$\mat{A}$, an $n\times n$ $s$-quasiseparable left triangular matrix:  $\mat{A}=\text{Left}(\mat{C}^\mat{A}\mat{R}^\mat{A}\mat{E}^\mat{A})$}
     \Require{$\mat{B}$, an $n\times t$ matrix}
     \Ensure{$\mat{C} \leftarrow \mat{A} \mat{B}$, an $n\times t$ dense tall and skinny
     matrix}
\For{$j=1\dots N-1$}
     \State  $\mat{X}_j \leftarrow \mat{E}^\mat{A}_j\mat{B}_j$ \label{step:ptwiseprod}
\Comment{ in a compact representation $ \mat{X}_j = \mat{D}_j^\mat{X}+\mat{T}^\mat{E}\mat{S}_j^\mat{X}$}
\State \label{step:merge} $\mat{Y}_j \leftarrow \mat{R}^\mat{A}
\mat{X}_j$\Comment{expand   $\mat{D}_j^\mat{X}+\mat{T}^\mat{E}\mat{S}_j^\mat{X}$ and apply the permutation
  $\mat{R}^\mat{A}$}
\EndFor
\State \label{step:prefix} compute all partial sums of these blocks: $\mat{Z}_i = \sum_{j=1}^{N-i}\mat{Y}_j$;
\State \label{step:leftmult} apply $\mat{C}_i^\mat{A}$ to the left: $\mat{V}_i = \mat{C}^\mat{A}_i \mat{Z}_i$;
\State \label{step:trailing} add the trailing term, $\mat{C}_i = \mat{V}_i   +
\text{Left}\left(\mat{C}_i^\mat{A}\mat{R}^\mat{A}\mat{E}^\mat{A}_{N - i + 1})\mat{B}_{N - i + 1}\right)$.
\State \Return $\mat{C}\leftarrow 
\begin{bmatrix}  \mat{C}_1 \\ \vdots \\ \mat{C}_{N}\end{bmatrix}
$
\end{algorithmic}
\end{algorithm}

In the compact Bruhat representation, the row echelon form $\mat{E}^\mat{A}$ is stored
in the form $\mat{E}^\mat{A} = \mat{D}^\mat{E} + \mat{T}^\mat{E} \mat{S}^\mat{E}$ where $\mat{D}$ and
$\mat{S}$ are block diagonal with blocks of dimension $s\times k_j$ where
$k_j\geq s$.

\begin{description}

\item[Step~\ref{step:ptwiseprod}] reduces to computing
$\mat{D}_j^\mat{X} = \mat{D}_j^\mat{E} \mat{B}_j$ and 
$\mat{S}_j^\mat{X} = \mat{S}_j^\mat{E} \mat{B}_j$ such that 
\begin{equation}
\label{eq:Xsnake}
  \mat{X}_j = \mat{D}_j^\mat{X}+\mat{T}^\mat{E}\mat{S}_j^\mat{X}.
\end{equation}
Each of these products requires $O(k_jst^{\omega-2})$ field operations, hence
Step~\ref{step:ptwiseprod} costs $O(nst^{\omega-2})$ field operations. After Step~\ref{step:ptwiseprod}, the
matrix~$\mat{X}$ is stored in a compact representation, given by
equation~\eqref{eq:Xsnake}, requiring only $O(nt)$ space.

\item[Step~\ref{step:merge}] does not involve any field operation as the multiplication on
the left by $\mat{T}^\mat{E}$ and the final sum act on matrices of
non-overlapping support. The overall amount of data being copied is linear in
the number of non-zero elements: $O(nt)$.

\item[Step~\ref{step:prefix}] can be achieved by computing the prefix sum of the
$\mat{Y}_i$'s: $\mat{Z}_1 = \mat{Y}_1$ and $\mat{Z}_i = \mat{Z}_{i-1}+\mat{Y}_i$. 
Each step involves $O(st)$ additions (the number of non zero elements in
$\mat{Y}_i$), hence Step~\ref{step:prefix} costs $O(nt)$ field operations.

\item[Step~\ref{step:leftmult}] is a sequence of $N$ products of an $s\times r$ 
matrix $\mat{C}^\mat{A}_i =\mat{D}_i^\mat{C} + \mat{S}_i^\mat{C}\mat{T}^\mat{C}$ 
by an $r\times r$ matrix $\mat{Z}_i$. As both $\mat{D}_i^\mat{C}$ and
$\mat{S}_i^\mat{C}$ have only $s$ continuous non-zero columns, each of these
product costs $O(s^2t^{\omega-2})$ and the overall cost is $O(nst^{\omega-2})$.

\item[Step~\ref{step:trailing}] is achieved by computing the $s\times s$ factor
$\text{Left}(\mat{C}_i^\mat{A}\mat{R}^\mat{A}\mat{E}^\mat{A}_{N -  i + 1})$ explicitly
in $O(s^\omega)$, and then applying it to $\mat{B}_{N - i + 1}$ in $O(s^2t^{\omega-2})$.

\end{description}

Overall the cost of algorithm~\ref{alg:LeftCBxTS} is $O(nts^{\omega-2})$ field operations.

\begin{corollary}
  An $s$-quasiseparable matrix in Compact Bruhat representation can be
  multiplied
  \begin{enumerate}[a.]
  \item by a vector in time $\GO{ns}$
  \item by a dense $n\times m$ matrix in time $\GO{s^{\omega-2}nm}$.
  \item by a another $s$-quasiseparable matrix matrix in time $\GO{s^{\omega-2}n^2}$.
  \end{enumerate}
\end{corollary}

\begin{proof}\ 
 \begin{enumerate}[a.]
   \item
  Specializing this \texttt{LeftCBxTS} algorithm with $t=1$ yields an algorithm for
multiplying by vector in time $\GO{ns}$.
\item
Splitting the dense matrix in $\lceil \frac{m}{s}\rceil$ slices and
applying \texttt{LeftCBxTS} on each of them takes
$\GO{s^{\omega-1}n \lceil\frac{m}{s}\rceil}=\GO{s^{\omega-2}nm}$.
\item Expanding one of the two matrices into a dense representation and
  multiplying it to the other one takes $\GO{s^{\omega-2}n^2}$.
 \end{enumerate}
\end{proof}

The last item in the corollary improves over the complexity of multiplying two
dense matrices in \GO{n^\omega}.
However, the result being itself a $2s$-quasiseparable matrix, it could be
presented in a Compact Bruhat representation. Hence the target cost for this operation
is far below: \GO{s^{\omega-1}n} since both input and output have size \GO{sn}.
Applying similar techniques as in Algorithm~\ref{alg:LeftCBxTS}, we could only
produce the output as two terms of the form $\LTP(\mat{L}\mat{R})$ where
$\mat{L}$ and $\mat{R}^T$ are 
$n\times(\text{rank}(\mat{A})+\text{rank}(\mat{B}))$ in time \GO{s^{\omega-1}n}, but
we were unable to perform the compression to a Compact Bruhat representation
within this target complexity for the moment.

{\small
  \bibliographystyle{elsarticle-harv}
\bibliography{qsprod}

\begin{thebibliography}{33}
\expandafter\ifx\csname natexlab\endcsname\relax\def\natexlab#1{#1}\fi
\expandafter\ifx\csname url\endcsname\relax
  \def\url#1{\texttt{#1}}\fi
\expandafter\ifx\csname urlprefix\endcsname\relax\def\urlprefix{URL }\fi

\bibitem[{Bini and Pan(1994)}]{BiPa94}
Bini, D., Pan, V., 1994. Polynomial and Matrix Computations, Volume 1:
  Fundamental Algorithms. Birkhauser, Boston.

\bibitem[{Boito et~al.(2016)Boito, Eidelman, and Gemignani}]{BEG16}
Boito, P., Eidelman, Y., Gemignani, L., 2016. Implicit {QR} for companion-like
  pencils. Math. of Computation 85~(300), 1753--1774.
\newline\urlprefix\url{http://www.ams.org/mcom/2016-85-300/S0025-5718-2015-03020-8/}

\bibitem[{Bostan et~al.(2008)Bostan, Jeannerod, and Schost}]{BJS08}
Bostan, A., Jeannerod, C.-P., Schost, E., Nov. 2008. Solving structured linear
  systems with large displacement rank. Theoretical Computer Science
  407~(1–3), 155--181.
\newline\urlprefix\url{http://www.sciencedirect.com/science/article/pii/S0304397508003940}

\bibitem[{Bruhat(1956)}]{Bru56}
Bruhat, F., 1956. Sur les représentations induites des groupes de {Lie}.
  Bulletin de la Société Mathématique de France 84, 97--205.
\newline\urlprefix\url{http://eudml.org/doc/86911}

\bibitem[{Carrier et~al.(1988)Carrier, Greengard, and Rokhlin}]{CGR88}
Carrier, J., Greengard, L., Rokhlin, V., Jul. 1988. A {Fast} {Adaptive}
  {Multipole} {Algorithm} for {Particle} {Simulations}. SIAM Journal on
  Scientific and Statistical Computing 9~(4), 669--686.
\newline\urlprefix\url{http://epubs.siam.org/doi/abs/10.1137/0909044}

\bibitem[{Chan(1987)}]{Chan87}
Chan, T.~F., Apr. 1987. Rank revealing {QR} factorizations. Linear Algebra and
  its Applications 88, 67--82.
\newline\urlprefix\url{http://www.sciencedirect.com/science/article/pii/0024379587901030}

\bibitem[{Chandrasekaran et~al.(2005)Chandrasekaran, Dewilde, Gu, Pals, Sun,
  van~der Veen, and White}]{CDGPSVW05}
Chandrasekaran, S., Dewilde, P., Gu, M., Pals, T., Sun, X., van~der Veen, A.,
  White, D., Jan. 2005. Some {Fast} {Algorithms} for {Sequentially}
  {Semiseparable} {Representations}. SIAM Journal on Matrix Analysis and
  Applications 27~(2), 341--364.
\newline\urlprefix\url{http://epubs.siam.org/doi/abs/10.1137/S0895479802405884}

\bibitem[{Chandrasekaran et~al.(2006)Chandrasekaran, Gu, and Pals}]{CGP06}
Chandrasekaran, S., Gu, M., Pals, T., Jan. 2006. A {Fast} {ULV} {Decomposition}
  {Solver} for {Hierarchically} {Semiseparable} {Representations}. SIAM Journal
  on Matrix Analysis and Applications 28~(3), 603--622.
\newline\urlprefix\url{http://epubs.siam.org/doi/abs/10.1137/S0895479803436652}

\bibitem[{Chandrasekaran and Ipsen(1994)}]{ChIp94}
Chandrasekaran, S., Ipsen, I., Apr. 1994. On {Rank}-{Revealing}
  {Factorisations}. SIAM Journal on Matrix Analysis and Applications 15~(2),
  592--622.
\newline\urlprefix\url{http://epubs.siam.org/doi/abs/10.1137/S0895479891223781}

\bibitem[{Delvaux and Van~Barel(2007)}]{DVB07}
Delvaux, S., Van~Barel, M., Nov. 2007. A {Givens}-{Weight} {Representation} for
  {Rank} {Structured} {Matrices}. SIAM J. on Matrix Analysis and Applications
  29~(4), 1147--1170.
\newline\urlprefix\url{http://epubs.siam.org/doi/abs/10.1137/060654967}

\bibitem[{Dumas et~al.(2013)Dumas, Pernet, and Sultan}]{DPS13}
Dumas, J.-G., Pernet, C., Sultan, Z., 2013. Simultaneous computation of the row
  and column rank profiles. In: Kauers, M. (Ed.), Proc. ISSAC'13. ACM Press,
  pp. 181--188.

\bibitem[{Dumas et~al.(2015)Dumas, Pernet, and Sultan}]{DPS15}
Dumas, J.-G., Pernet, C., Sultan, Z., 2015. Computing the rank profile matrix.
  In: Proc ISSAC'15. ACM, New York, NY, USA, pp. 149--156, distinguished paper
  award.
\newline\urlprefix\url{http://doi.acm.org/10.1145/2755996.2756682}

\bibitem[{Dumas et~al.(2016)Dumas, Pernet, and Sultan}]{DPS16}
Dumas, J.-G., Pernet, C., Sultan, Z., 2016. Fast computation of the rank
  profile matrix and the generalized {B}ruhat decomposition. Journal of
  Symbolic Computation.

\bibitem[{Eidelman and Gohberg(1999)}]{EiGo99}
Eidelman, Y., Gohberg, I., Sep. 1999. On a new class of structured matrices.
  Integral Equations and Operator Theory 34~(3), 293--324.
\newline\urlprefix\url{http://link.springer.com/article/10.1007/BF01300581}

\bibitem[{Eidelman and Gohberg(2005)}]{EiGo05}
Eidelman, Y., Gohberg, I., Dec. 2005. On generators of quasiseparable finite
  block matrices. CALCOLO 42~(3-4), 187--214.
\newline\urlprefix\url{http://link.springer.com/article/10.1007/s10092-005-0102-4}

\bibitem[{Eidelman et~al.(2005)Eidelman, Gohberg, and Olshevsky}]{EGO05}
Eidelman, Y., Gohberg, I., Olshevsky, V., 2005. The {QR} iteration method for
  hermitian quasiseparable matrices of an arbitrary order. Linear Algebra and
  its Applications 404, 305 -- 324.
\newline\urlprefix\url{http://www.sciencedirect.com/science/article/pii/S0024379505001369}

\bibitem[{Gohberg et~al.(1985)Gohberg, Kailath, and Koltracht}]{GKK85}
Gohberg, I., Kailath, T., Koltracht, I., Nov. 1985. Linear complexity
  algorithms for semiseparable matrices. Integral Equations and Operator Theory
  8~(6), 780--804.
\newline\urlprefix\url{http://link.springer.com/article/10.1007/BF01213791}

\bibitem[{Hwang et~al.(1992)Hwang, Lin, and Yang}]{HWY92}
Hwang, T.-M., Lin, W.-W., Yang, E.~K., Oct. 1992. Rank revealing {LU}
  factorizations. Linear Algebra and its Applications 175, 115--141.
\newline\urlprefix\url{http://www.sciencedirect.com/science/article/pii/002437959290305T}

\bibitem[{Jeannerod et~al.(2013)Jeannerod, Pernet, and Storjohann}]{JPS13}
Jeannerod, C.-P., Pernet, C., Storjohann, A., 2013. Rank-profile revealing
  {G}aussian elimination and the {CUP} matrix decomposition. J. Symbolic
  Comput. 56, 46--68.

\bibitem[{Kailath et~al.(1979)Kailath, Kung, and Morf}]{KKM79}
Kailath, T., Kung, S.-Y., Morf, M., Apr. 1979. Displacement ranks of matrices
  and linear equations. Journal of Mathematical Analysis and Applications
  68~(2), 395--407.
\newline\urlprefix\url{http://www.sciencedirect.com/science/article/pii/0022247X79901240}

\bibitem[{Le~Gall(2014)}]{LeG14}
Le~Gall, F., 2014. Powers of tensors and fast matrix multiplication. In:
  Proceedings of the 39th International Symposium on Symbolic and Algebraic
  Computation. ISSAC '14. ACM, New York, NY, USA, pp. 296--303.
\newline\urlprefix\url{http://doi.acm.org/10.1145/2608628.2608664}

\bibitem[{Malaschonok(2010)}]{Mal10}
Malaschonok, G.~I., 2010. Fast generalized {Bruhat} decomposition. In: CASC'10.
  Vol. 6244 of LNCS. Springer-Verlag, Berlin, Heidelberg, pp. 194--202.

\bibitem[{Manthey and Helmke(2007)}]{MH07}
Manthey, W., Helmke, U., 2007. {Bruhat} canonical form for linear systems.
  Linear Algebra and its Applications 425~(2–3), 261 -- 282, special Issue in
  honor of Paul Fuhrmann.

\bibitem[{Pan(2000)}]{Pan00}
Pan, C.-T., Sep. 2000. On the existence and computation of rank-revealing {LU}
  factorizations. Linear Algebra and its Applications 316~(1–3), 199--222.
\newline\urlprefix\url{http://www.sciencedirect.com/science/article/pii/S0024379500001208}

\bibitem[{Pan(1990)}]{Pan90}
Pan, V., 1990. On computations with dense structured matrices. Mathematics of
  Computation 55~(191), 179--190.

\bibitem[{Pernet(2016)}]{Per16}
Pernet, C., 2016. Computing with quasiseparable matrices. In: Proc. ISSAC'16.
  ACM, pp. 389--396,
  \href{https://hal.archives-ouvertes.fr/hal-01264131}{\texttt{hal-01264131}}.

\bibitem[{Sheng et~al.(2007)Sheng, Dewilde, and Chandrasekaran}]{SDC07}
Sheng, Z., Dewilde, P., Chandrasekaran, S., 2007. Algorithms to {Solve}
  {Hierarchically} {Semi}-separable {Systems}. In: Alpay, D., Vinnikov, V.
  (Eds.), System {Theory}, the {Schur} {Algorithm} and {Multidimensional}
  {Analysis}. No. 176 in Operator {Theory}: {Advances} and {Applications}.
  Birkhäuser Basel, pp. 255--294, dOI: 10.1007/978-3-7643-8137-0\_5.
\newline\urlprefix\url{http://link.springer.com/chapter/10.1007/978-3-7643-8137-0_5}

\bibitem[{Strassen(1969)}]{Str69}
Strassen, V., 1969. {Gaussian} elimination is not optimal. Numerische
  Mathematik 13, 354--356.

\bibitem[{{The LinBox Group}(2016)}]{linbox:2016}
{The LinBox Group}, 2016. {LinBox}: Linear algebra over black-box matrices.
  v1.4.1 Edition, \url{http://linalg.org/}.

\bibitem[{Tyrtyshnikov(1997)}]{Tyr97}
Tyrtyshnikov, E., 1997. Matrix {Bruhat} decompositions with a remark on the
  {QR} ({GR}) algorithm. Linear Algebra and its Applications 250, 61 -- 68.

\bibitem[{Vandebril et~al.(2005)Vandebril, Barel, Golub, and
  Mastronardi}]{VBGM05}
Vandebril, R., Barel, M.~V., Golub, G., Mastronardi, N., 2005. A bibliography
  on semiseparable matrices. CALCOLO 42~(3), 249--270.
\newline\urlprefix\url{http://dx.doi.org/10.1007/s10092-005-0107-z}

\bibitem[{Vandebril et~al.(2007)Vandebril, Van~Barel, and Mastronardi}]{VVM07}
Vandebril, R., Van~Barel, M., Mastronardi, N., 2007. Matrix computations and
  semiseparable matrices: linear systems. Vol.~1. The Johns Hopkins University
  Press.

\bibitem[{Xia et~al.(2010)Xia, Chandrasekaran, Gu, and Li}]{XCSGL10}
Xia, J., Chandrasekaran, S., Gu, M., Li, X.~S., Dec. 2010. Fast algorithms for
  hierarchically semiseparable matrices. Numerical Linear Algebra with
  Applications 17~(6), 953--976.
\newline\urlprefix\url{http://onlinelibrary.wiley.com/doi/10.1002/nla.691/abstract}

\end{thebibliography}
}
\end{document}